\documentclass[journal,10pt,draftclsnofoot,onecolumn]{IEEEtran}

\newcommand{\diag}{\mathop{\mathrm{diag}}\nolimits}

\usepackage{cite}
\ifCLASSINFOpdf
 \else
 \fi
 \usepackage{graphicx,color}
\usepackage[cmex10]{amsmath}
\usepackage{ amssymb }
\usepackage{bm} 
\usepackage{amsthm}
 \newtheorem{definition}{Definition}
 \newtheorem{property}{Property}
 \newtheorem{problem}{Problem}
 
 \newtheorem{claim}{Claim}
 \newtheorem{remark}{Remark}

\usepackage{algorithm}
\usepackage{algpseudocode}
\usepackage[caption=false,font=footnotesize]{subfig}

\begin{document}
\title{An Efficient High-Dimensional Sparse Fourier Transform\footnote{This manuscript has been submitted to IEEE Transactions on Aerospace and Electronic Systems for reviewing on Sept. 1st, 2016.}}

\author{{Shaogang Wang,~\IEEEmembership{Student Member,~IEEE},
Vishal M. Patel,~\IEEEmembership{Senior Member,~IEEE},
and Athina Petropulu,~\IEEEmembership{Fellow,~IEEE}} 
\thanks{Shaogang Wang, Vishal M. Patel and Athina Petropulu are with the department of Electrical and Computer Engineering at Rutgers University, 
   Piscataway, NJ USA. Email: \{shaogang.wang, vishal.m.patel, athinap\}@rutgers.edu. }
}


\maketitle


\begin{abstract}
We propose RSFT, which is an extension of  the one dimensional Sparse Fourier Transform algorithm to higher dimensions in a way that it can be applied to real, noisy data. The RSFT allows for off-grid frequencies. Furthermore, by incorporating  Neyman-Pearson detection, the frequency detection stages in RSFT do not require  knowledge of the exact sparsity of the signal and are more robust to noise.  We analyze the asymptotic
performance of RSFT, and study the computational complexity versus the worst case signal SNR tradeoff. We show that by choosing the proper parameters, the optimal tradeoff can be achieved. We discuss the application of RSFT on short range ubiquitous radar signal processing, and demonstrate its feasibility via simulations.
\end{abstract}

\begin{IEEEkeywords}
Array signal processing, sparse Fourier transform, detection and estimation,  radar signal processing.
\end{IEEEkeywords}

\section{Introduction}

\IEEEPARstart{H}{igh} dimensional FFT (N-D FFT) is used in many applications in which a multidimensional Discrete Fourier Transform (DFT) is needed, such as radar and imaging. However, its complexity increases with an increase in dimension, leading to costly hardware and low speed of system reaction.  The recently developed Sparse Fourier Transform (SFT) \cite{hassanieh2012nearly, Hassanieh:2012:SPA:2095116.2095209, ong2015fast, gilbert2014recent, ghazi2013sample, indyk2014sample}, designed for signals that have a small number of frequencies enjoys low complexity, and thus is ideally suited in big data scenarios \cite{hassanieh2012faster, hassanieh2014ghz, shi2014light}. One such application is radar target detection. The radar returns are typically sparse  in the target space, i.e., the number of targets is far less than the number of resolutions cells. This observation motivates the application of SFT in  radar signal processing.

The current literature on  multi-dimensional extensions of the SFT include \cite{rauhsparse, ghazi2013sample, ong2015fast, indyk2014sample}. Those works mainly address sample complexity, i.e., using the least number of time domain samples to reconstruct the signal frequencies. In order to detect the significant frequencies in an approximately sparse settings, i.e., signal corrupted by additive noise, the aforementioned methods assume knowing the exact sparsity, and compare the frequency amplitude with a predefined threshold. However, in many real applications, the exact signal sparsity may be either unknown or subject to change. For example, in the radar case, the number of targets to be detected is typically unknown and usually varies from time to time. Also, setting up an ideal threshold for detection in noisy cases is not trivial, since it relates to the tradeoff between probability of detection and false alarm rate. However, those issues have not been studied in the SFT literature.

One of the constrains in the aforementioned SFT algorithms is the assumption that  the signal discrete  frequencies are all on-grid. In reality, however, the signal frequencies, when discretized and depending on the grid size can fall between grid points. The consequence of off-grid frequencies is leakage to other frequency bins, which essentially reduces the  sparsity  of the signal. To refine the estimation of off-grid frequencies starting from the initial SFT-based estimates,\cite{shi2014light} proposed a gradient descent method. However, the method of \cite{shi2014light} has to enumerate all possible directions of the gradient for each frequency and compute the approximation error for each guessed direction,  which increases the computational complexity. Moreover, like the other aforementioned SFT methods, the thresholding for frequency detection is not clear in  \cite{shi2014light}. The off-grid frequency problem in the context of SFT was also studied   in \cite{boufounos2012s}, where it was assumed that in the frequency domain the signal and the noise are well separated by predefined gaps. However, this is a restrictive assumption that limits the applicability of the work.

In this paper we setup a SFT-based framework for sparse signal detection in a high dimensional frequency domain and propose a new algorithm, namely, the Realistic Sparse Fourier Transform (RSFT) which addresses the shortcomings discussed above. This paper makes the following contributions.
\begin{enumerate}
\item The RSFT algorithm does not require knowledge of the number of frequencies to be estimated. Also, it does not need the frequencies to be on-grid and does not require signal and noise to be separated in the frequency domain.  Our method reduces the leakage effect from off-grid frequencies by applying a window on the input time domain data (see Section \ref{sec:leakage} for details). The design of this window trades off frequency resolution, leakage suppression and computational efficiency. We shall point out that, unlike the work of  \cite{shi2014light} that recovers the exact off-grid frequency locations, our work aims to recover \emph{grided locations} of off-grid frequencies with less amount of computation.

\item We extend the RSFT into an arbitrary fixed high dimension, so that it can replace the N-D FFT in sparse settings and thus enable computational savings.
\item  We put RSFT into a Neyman-Pearson (NP) detection framework. Based on the signal model and other design specifications, we give the (asymptoticly) optimal  thresholds for the two detection stages of the RSFT (see 
Section \ref{sec:optimal} for details). Since the output of the first stage of detection serves as the input of the second stage, the two detection stages are interconnected. The detection thresholds are jointly found by formulating and solving an optimization problem, with its objective function minimizing the worst case Signal to Noise Ratio (SNR) (hence the system is more sensitive to weak signal), and its constrains connecting probability of detection and false alarm rate for both two stages.
\item We provide a quantitive measure of tradeoff between computational complexity and worst case signal SNR  for systems that use RSFT, which serves as a concrete design reference for system engineers.  
\item  We investigate the use of RSFT in multi-dimensional radar signal processing.
\end{enumerate}

A closely related technique to SFT is \emph{compressed sensing} (CS).  Compressed sensing-based methods  recover signals by exploiting their sparse features  \cite{cande2008introduction}. The application of CS methods  in MIMO radar is discussed intensively in \cite{chen2008compressed,yu2010mimo}.  At a high level, CS differs from SFT in that it assumes the signal can be sparsely represented by an overdetermined dictionary, and formulates the problem as an optimization problem with sparsity constraints ($l_0$ or $l_1$ norm). This problem is usually solved by convex optimization, which runs in polynomial time in $N$ for an $N$-dimensional signal. On the other hand, the SFT method finds an (approximately) \emph{sparse Fourier representation} of a signal in a mean square error sense ($l_2$ norm). The sample and the computational complexity of SFT is sub-linear to $N$, for a wide range of $N$-dimensional signals \cite{indyk2014sample}.   

A preliminary version of this work appeared in \cite{Wang1611:RSFT}.   A detailed analysis of the RSFT algorithm as well as extensive experimental results are extensions to \cite{Wang1611:RSFT}.

\subsection{Notation}
We use lower-case (upper-case) bold letters to denote vectors (matrices). $(\cdot)^T$ and $(\cdot)^H$ respectively denote the transpose and conjugate transpose of a matrix or a vector. $\|\cdot\|$ is Euclidean norm for a vector.   $[\mathbf{a}]_i$ is the $i_{th}$ element of vector $\mathbf{a}$.  All operations on indices in this paper are taken modulo $N$ denoted by $[\cdot]_N$. We use $\lfloor \cdot \rfloor$ to denote  rounding to floor. $[S]$ refers to the set of indices $\{0, . . . , S-1  \}$, and $[S]\setminus a$ is for eliminating element $a$ from set $[S]$. We use $\{0,1\}^B$ to denote the set of $B$-dimensional binary vectors.  We use $\diag (\cdot)$ to  denote forming a diagonal matrix from a vector and use $\mathbb{E}\{ \cdot \}$ to denote expectation. The DFT of signal ${\mathbf{s}}$ is denoted as $\hat{{\mathbf{s}}}$. We also assume that the signal length in each dimension is an integer power of 2.

This paper is organized as follows.   A brief background on the SFT algorithm is given in Section~\ref{sec:background}.  Details of the proposed RSFT algorithm are given in Section \ref{sec:algorithm}. Section \ref{sec:optimal} presents the derivation of the optimal threshold design for the RSFT algorithm. Then in Section \ref{sec:numerical} we provide some numerical results to verify the theoretical findings. An application of the  RSFT algorithm in radar signal processing is presented in Section \ref{sec:application} and finally, concluding remarks are made in Section \ref{sec:conclusion}.

\section{Preliminaries} \label{sec:background}
\subsection{Basic Techniques}
As opposed to the FFT which computes the coefficients of all $N$ discrete frequency components of an $N$-sample long signal, the SFT\cite{Hassanieh:2012:SPA:2095116.2095209}  computes only the $K$ frequency components of a $K$-sparse signal. Before outlining the SFT algorithm we provide some key definitions and properties, which is extracted and reformulated based on \cite{Hassanieh:2012:SPA:2095116.2095209} .

\begin{definition}
\textbf{(Permutation):} Define the transform $P_{\sigma, \tau}$ such that, given $\mathbf{x} \in \mathbb{C}^{N},\sigma, \tau \in [N]$, and $\sigma$ invertible 
\begin{equation} \label{eq:mod_inv}
\sigma \sigma^{-1} = [1]_N.
\end{equation}
Then, the following transformation is called \emph{permutation}
\begin{equation} \label{eq:permutation}
(P_{\sigma,\tau}x)_i = x_{\sigma i+\tau} .
\end{equation}
\end{definition}
The permutation has the following property.
\begin{property} \label{prop:permutation}
A modular reordering of the data in time domain results in a modular dilation and phase rotation in the frequency domain, i.e.,
\begin{equation} \label{eq:freqDilation}
\widehat{(P_{\sigma, \tau} x)}_{\sigma i} = \hat{x}_i  e^{-j\tau \frac{2\pi}{N}} .
\end{equation}
\end{property}

\begin{definition} \label{def:aliasing}
\textbf{(Aliasing):}
Let $\mathbf{x} \in \mathbb{C}^{N}, \mathbf{y} \in \mathbb{C}^{B}$, with $N,B$ powers of 2, and $B<N$. For $L = {N}/{B}$, a time-domain  aliased version of x is defined as
\begin{equation} \label{eq:aliasing}
y_i = \sum_{j=0}^{L-1} x_{i+Bj}, \quad i\in[B] .
\end{equation}
\end{definition}

\begin{property} \label{prop:aliasing}
Aliasing in time domain results in downsampling in frequency domain, i.e.,
\begin{equation} \label{eq:aliasing_prop}
\hat{y}_i = \hat{x}_{iL} .
\end{equation}
\end{property}

\begin{definition} \label{def:map}
\textbf{(Mapping):} Let $i, \sigma \in [N]$, where $\sigma$ satisfies (\ref{eq:mod_inv}). We define the mapping $\mathcal{M}(i, \sigma)$ such that 
\begin{equation} 
[B] \ni \mathcal{M}(i, \sigma) \equiv \lfloor \frac{B}{N}[i \sigma]_N \rfloor .
\end{equation}
\end{definition}

\begin{definition} \label{def:rev_map}
\textbf{(Reverse-mapping):} Let $\sigma^{-1} \in [N]$, $\sigma^{-1}$ satisfies  (\ref{eq:mod_inv}), and $ j \in [B]$. Define $\mathcal{R}(j, \sigma^{-1})$ a reverse-mapping such that
\begin{equation} 
\begin{split}
&\mathcal{R}(j, \sigma^{-1}) \equiv \{ [\sigma^{-1} u]_N \mid u \in \mathbb{S}\},
\end{split}
\end{equation}
where
\begin{equation} 
\begin{split}
\mathbb{S} \equiv \{v \in[N] \mid   j\frac{N}{B} \le v < (j+1)\frac{N}{B}]\}.
\end{split}
\end{equation}
\end{definition}

\subsection{SFT Algorithm}
At a high level, the SFT algorithm in \cite{Hassanieh:2012:SPA:2095116.2095209} runs two loops, namely the \emph{Location} loop and the \emph{Estimation} loop. The former finds the indices of the $K$ most significant frequencies from the input signal, while the latter estimates the corresponding Fourier coefficients. Here, we emphasize on \emph{Location} more than \emph{Estimation}, since the former is more relevant to the radar application that we consider. The \emph{Location} step provides frequency locations, which in the radar case is directly related to target parameters.   

In the \emph{Location} loop, a \emph{permutation} procedure reorders the input data in the time domain, causing the frequencies to also reorder. The permutation causes closely spaced frequencies to appear in well separated locations with high probability. Then, a flat-window\cite{Hassanieh:2012:SPA:2095116.2095209} is applied on the permuted signal for the purpose of extending  a single frequency  into a (nearly) boxcar, for a reason that will become apparent in the following. The windowed data are aliased, as in Definition \ref{def:aliasing}. The frequency domain equivalent of this aliasing is undersampling by $N/B$ (see Property \ref{prop:aliasing}).  The flat-window used at the previous step ensures that no peaks are lost due to the effective undersampling in the frequency domain. After this stage, a FFT of length $B$ is employed.

The permutation and the aliasing procedure effectively \emph{map}  the signal frequencies from $N$-dimensional space into a reduced $B$-dimensional space, where the first-stage-detection procedure  finds the significant frequencies' peaks, and then their indices are \emph{reverse mapped} into the original $N$-dimensional frequency space. However, the reverse mapping yields not only the true location of the significant frequencies, but also ${N}/{B}$ ambiguous locations for each frequency. To remove the ambiguity, multiple iterations of \emph{Location} with randomized permutation are performed. Finally, the second-stage-detection procedure locates the $K$ most significant frequencies from the accumulated data for each iteration.  More details about the SFT algorithm can be found in \cite{Hassanieh:2012:SPA:2095116.2095209}.

\section{The RSFT Algorithm} \label{sec:algorithm}

In this section we address some problems that have not been considered in the original SFT algorithm, namely the leakage from off-grid frequencies and optimal detection threshold design for both detection stages. Also, we extend the RSFT into arbitrary fixed high dimension. 

\subsection{Leakage Suppression for Off-grid Frequencies} \label{sec:leakage}
In real world applications, the  frequencies  are continuous and can take any value in $[0, 2\pi)$. When fitting a grid on these frequencies, leakage occurs from off-grid frequencies, which can diminish the sparsity of the signal. As the leakage due to strong frequency components can mask the contributions of weak frequency components,  it becomes difficult to determine the frequency domain peaks after permutation. (see Fig. \ref{fig:PreWinMotive} (c)).  To address this problem, we propose to multiply  the received time domain signal with a  window  before permutation. We call this procedure \emph{pre-permutation windowing}. The idea is to confine the leakage within a finite number of frequency bins, as illustrated in Fig. \ref{fig:PreWinMotive}.   

The choice of the pre-permutation window is determined by the required resolution, computational complexity, and degree of leakage suppression. Specifically, the side-lobe level should be lower than the noise level after windowing (see Fig. \ref{fig:Win_Permutation}). However, the larger the attenuation of the side-lobes, the wider the main lobe would be, thus lowering the frequency resolution. Meanwhile, a broader main lobe results in increased computational load, which will be discussed in Section \ref{seg:complexity}.

\begin{figure}[!t]
    \centering
    \subfloat[]{{\includegraphics[scale=0.2]{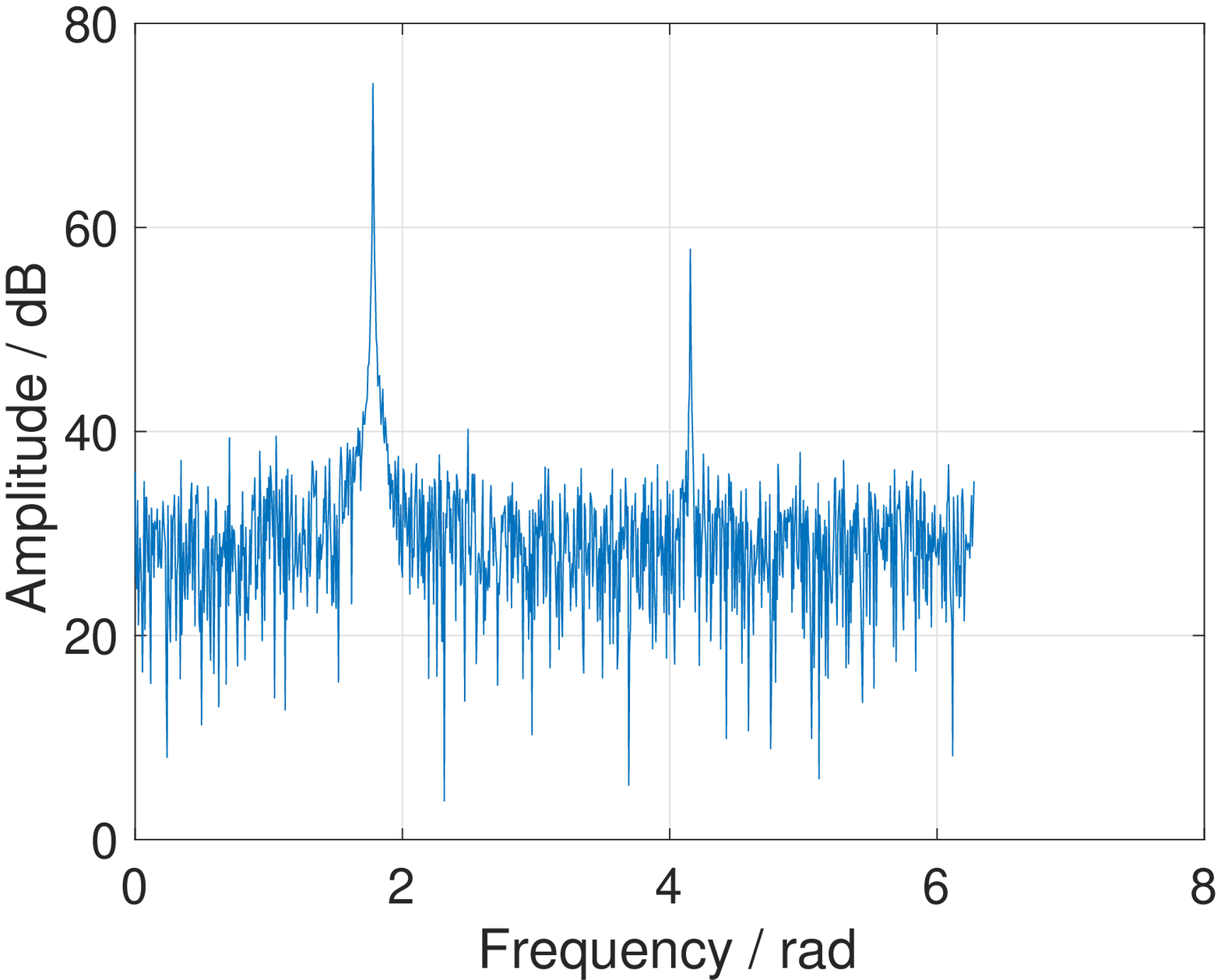} }}%
    \subfloat[]{{\includegraphics[scale=0.2]{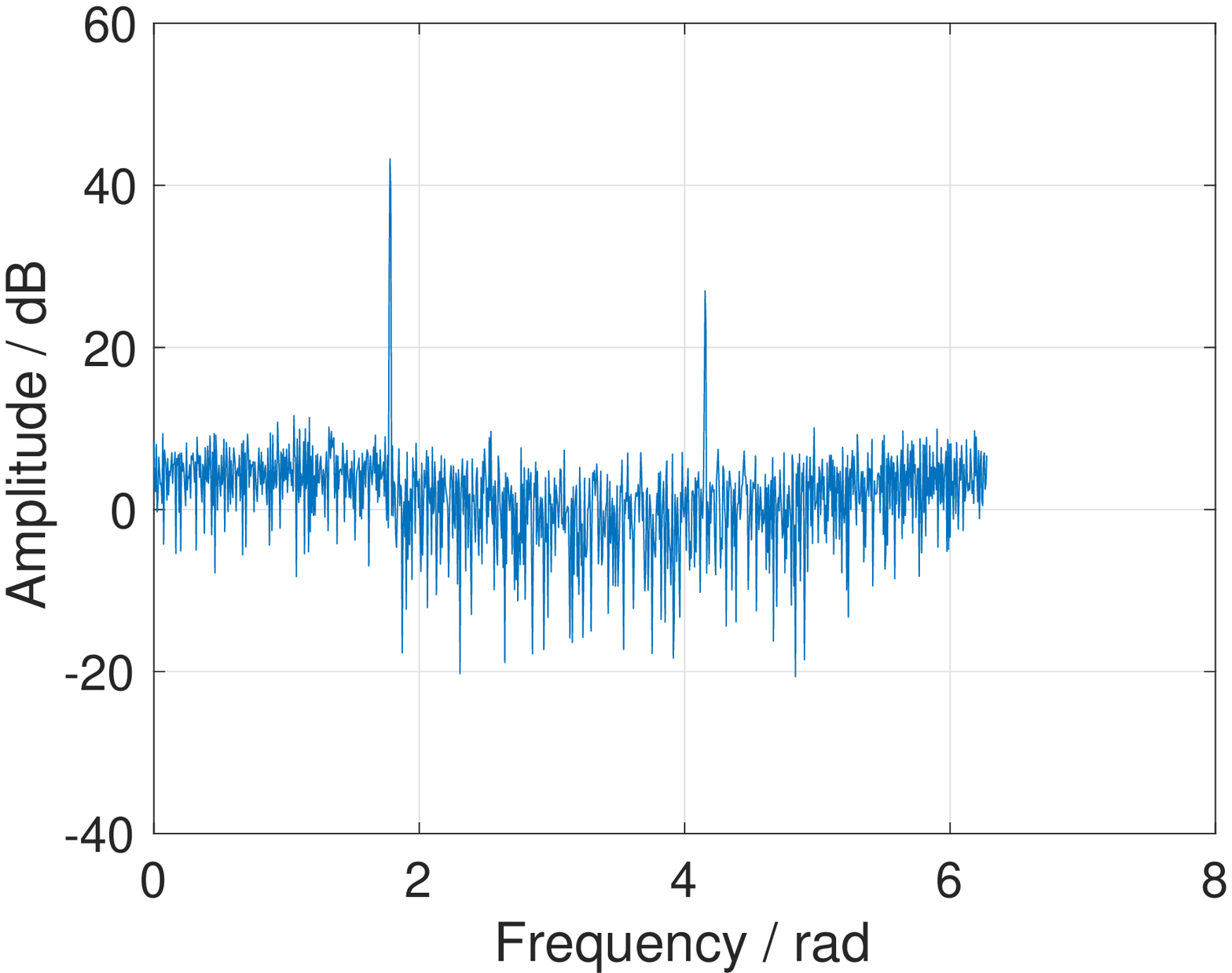} }}%
     \hfil
    \subfloat[]{{\includegraphics[scale=0.2]{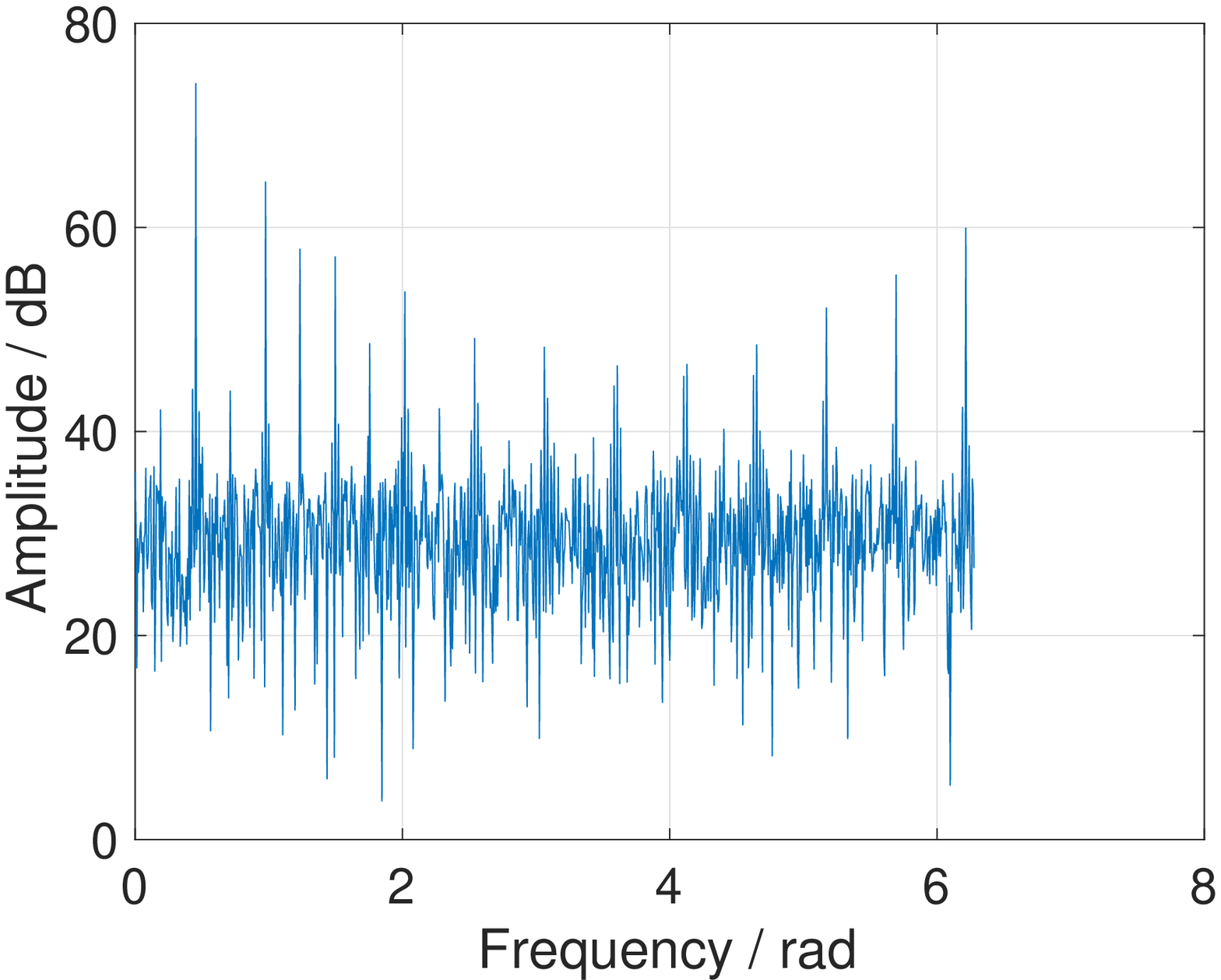} }}%
    \subfloat[]{{\includegraphics[scale=0.2]{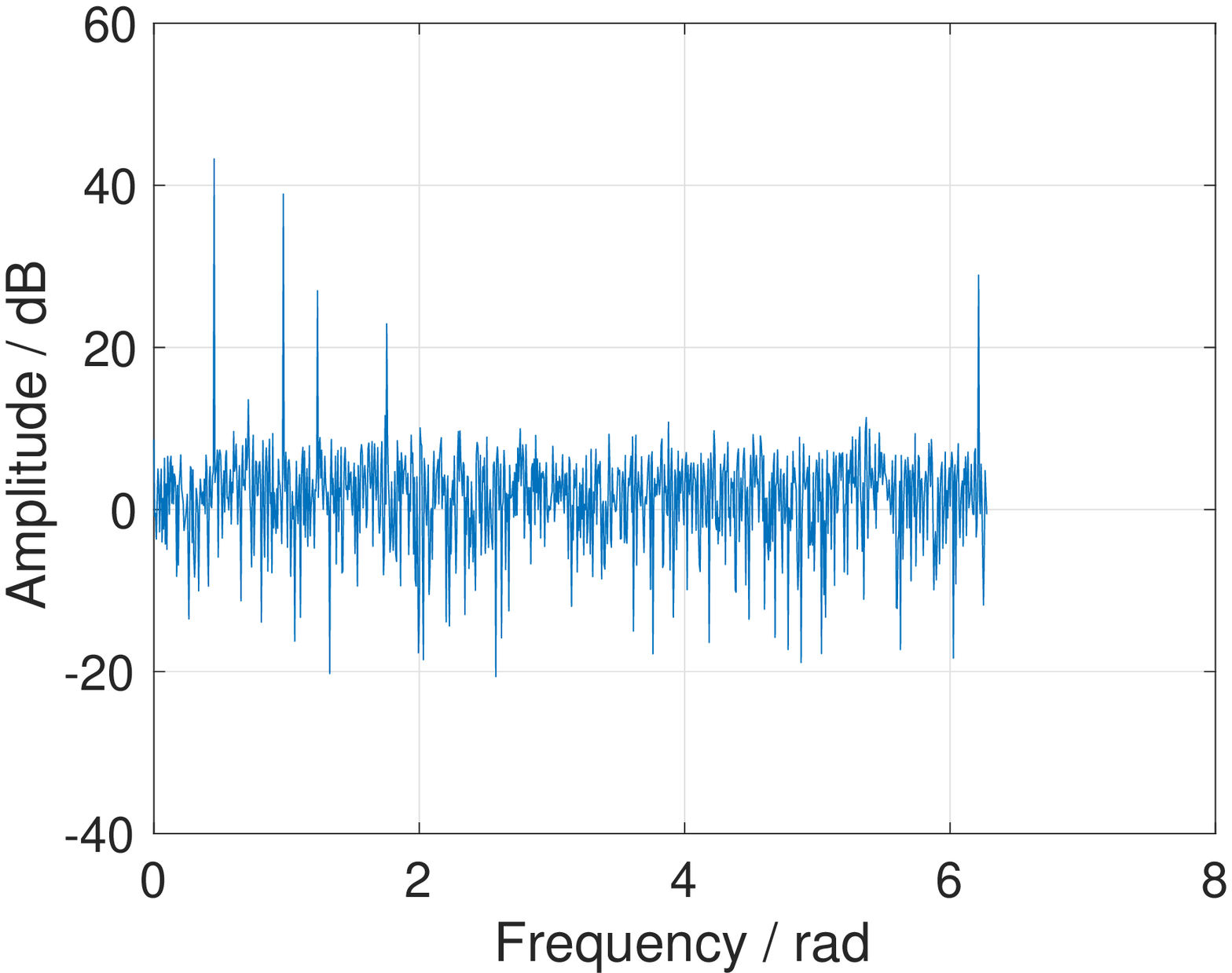} }}%
    \caption{\textbf{Effect of pre-permutation windowing.} The signal contains two significant frequency components, one of which is $15dB$ stronger than the other. A Dolph-Chebyshev window is applied to the time-domain signal. Windowed signal after permutation appears sparser in the frequency domain as compared to the permuted signal without windowing. (a) Spectrum of signal without windowing. (b) Spectrum of windowed signal. (c) Spectrum of signal without windowing after permutation. (d) Spectrum of windowed signal after permutation.}
    \label{fig:PreWinMotive}
\end{figure}

\subsection{Signal Detection Without Knowing the Exact Sparsity and Optimal Threshold Design}

In the SFT, detection of the significant frequencies is needed in two stages. With knowing the number of the significant frequencies and assuming they are all on-grid, the detection of the signal can be accomplished by finding the $K$ highest spectral amplitude values. 

In reality however, we usually do not have the knowledge of the exact sparsity, i.e., $K$. Moreover, even if we knew $K$, due to the leakage caused by the off-grid signals, the $K$ highest spectral peaks might not be the correct representation of the signal frequencies. Finally, the additive noise could generate false alarms, which would add more difficulties in signal detection.

In order to solve the detection problem, we propose to use NP detection in the  two stages of detection, which does not require knowing $K$. However, the optimal thresholds still rely on the number of significant frequencies as well as their SNR. In light of this, we provide an optimal design based on a bound for the sparsity and signal SNR level. The optimal threshold design is presented in Section \ref{sec:optimal}.

\subsection{High Dimensional Extensions} \label{sec:hd}
In the following, we elaborate on the high dimensional extension of the RSFT for its main stages.
\subsubsection{Windowing} In the pre-permutation windowing and the flat-windowing stages, the window for each dimension is designed separately. After that, the high dimension widow is generated by combining each 1-D window. For instance, in the 2-D case, assuming that $\mathbf{w}_x$ and $\mathbf{w}_y$ are the two windows in the $x$ and $y$ dimension, respectively, a 2D window can be computed as 
\begin{equation}
\mathbf{W}_{xy} = \mathbf{w}_x \mathbf{w}_y^H.
\end{equation} 

Fig. \ref{fig:2dwin} shows a compound 2-D window which is a combination of a 64-point and a 1024-point Dolph-Chebyshev window, both of which has $60dB$ attenuation for the side-lobes in frequency domain.  A Dolph-Chebyshev window allows us to trade off frequency resolution and side-lobe attenuation easly. Other windows sharing the same flexibility includes Gaussian window, Kaiser-Bessel window, Blackman-Harris window, etc., see \cite{harris1978use} for detail.  We apply those windows on the data by point-wise multiplications. 
 
 \begin{figure}[!t]
    \centering
    \subfloat[]{{\includegraphics[scale=0.22]{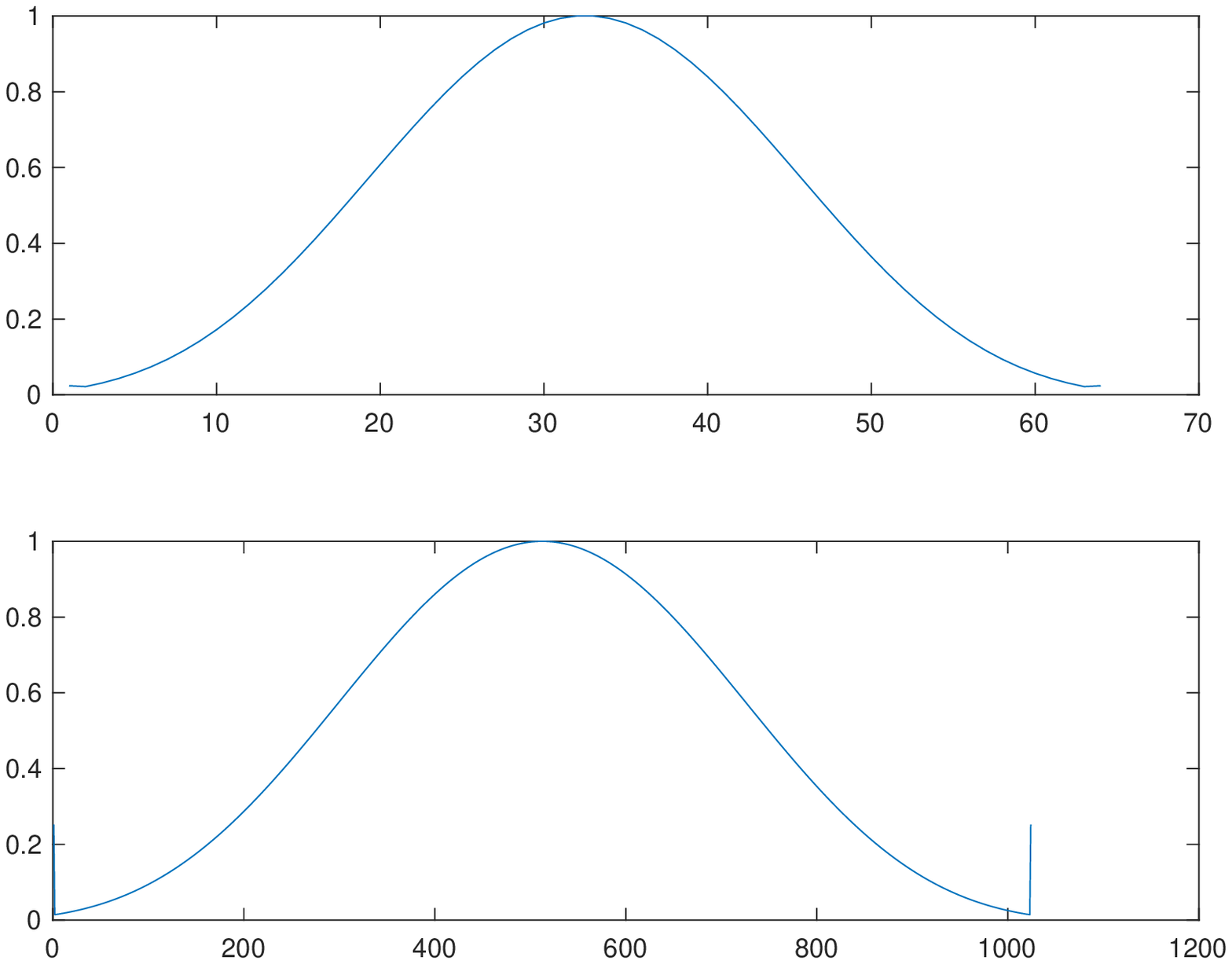} }}%
    \subfloat[]{{\includegraphics[scale=0.23]{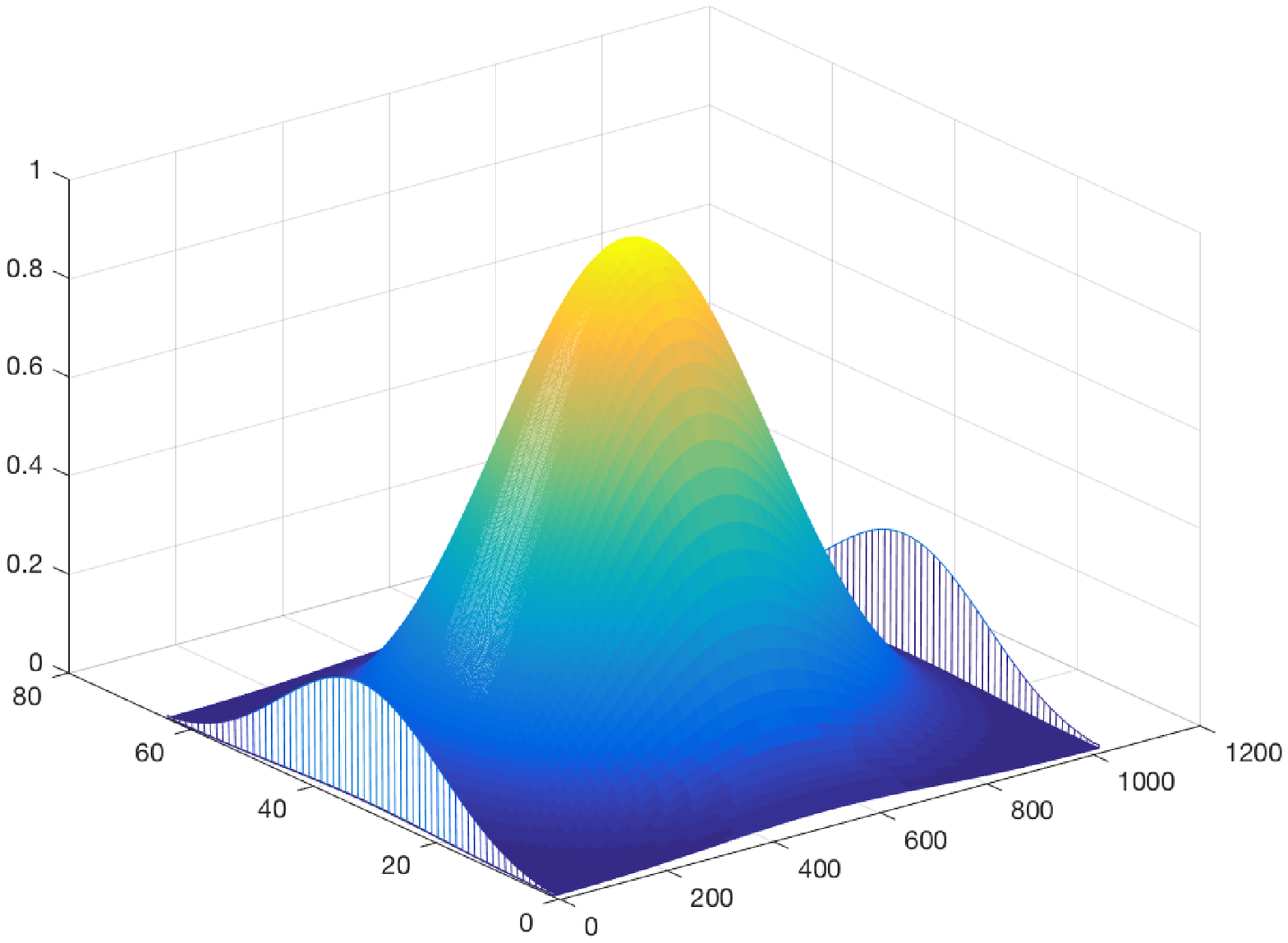} }}%
    \caption{\textbf{Compound Window in 2-D.} (a) Top: a $64$-points Dolph-Chebyshev window; bottom: a $1024$-points Dolph-Chebyshev window. (b) The  2D window.}
    \label{fig:2dwin}
\end{figure}

\subsubsection{Permutation} The permutation parameters are generated for each dimension in a random way according to  (\ref{eq:permutation}). Then, we carry the permutation on each dimension sequentially. An example for the 2-D case is illustrated in Fig. \ref{fig:Permutation}.

\begin{figure}[!t]
    \centering
    \subfloat[]{{\includegraphics[scale=0.31]{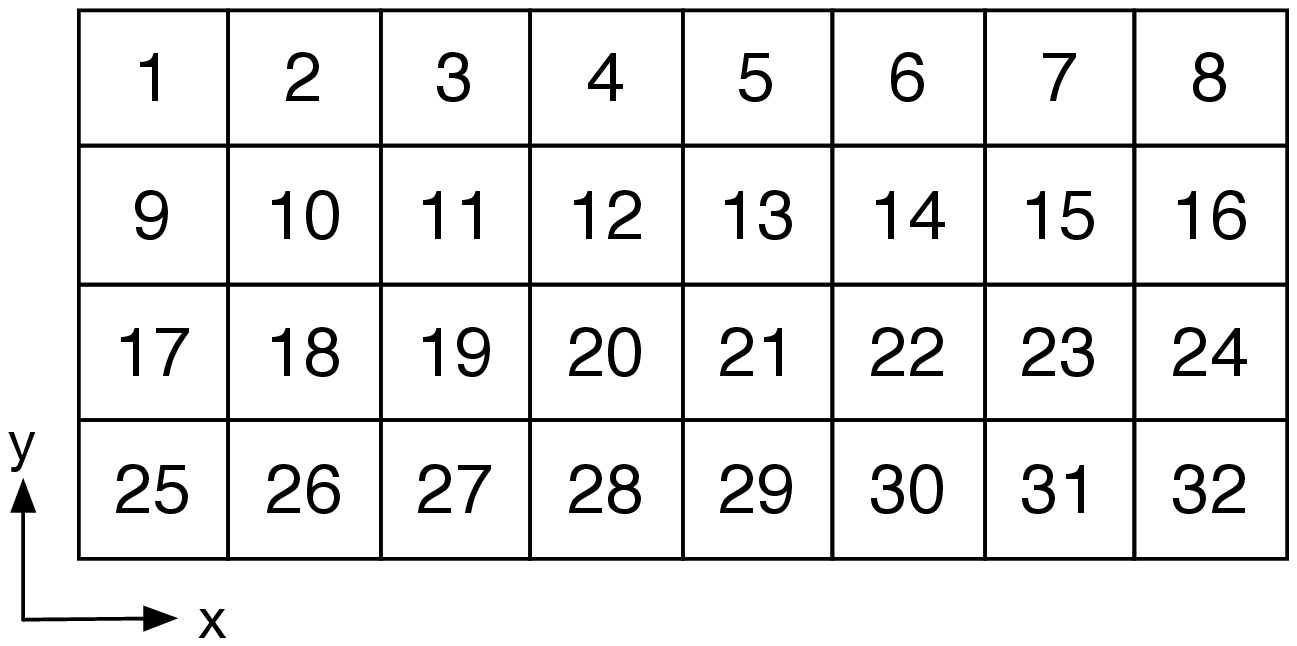} }}%
    \subfloat[]{{\includegraphics[scale=0.31]{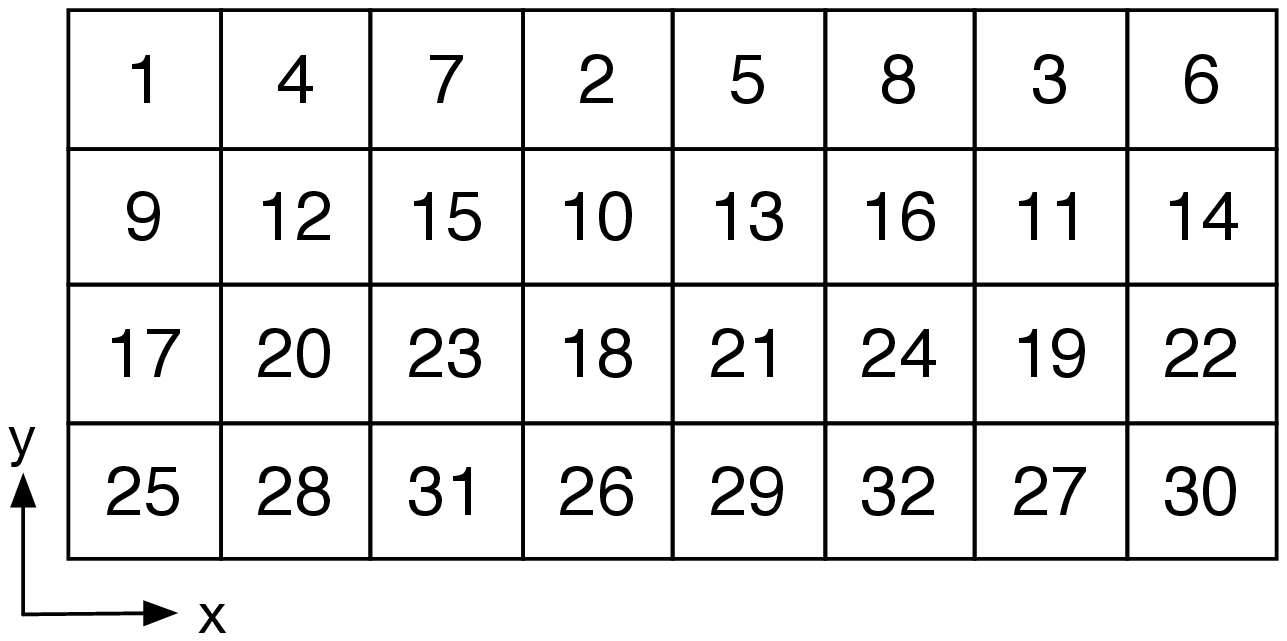} }}%
     \hfil
    \subfloat[]{{\includegraphics[scale=0.31]{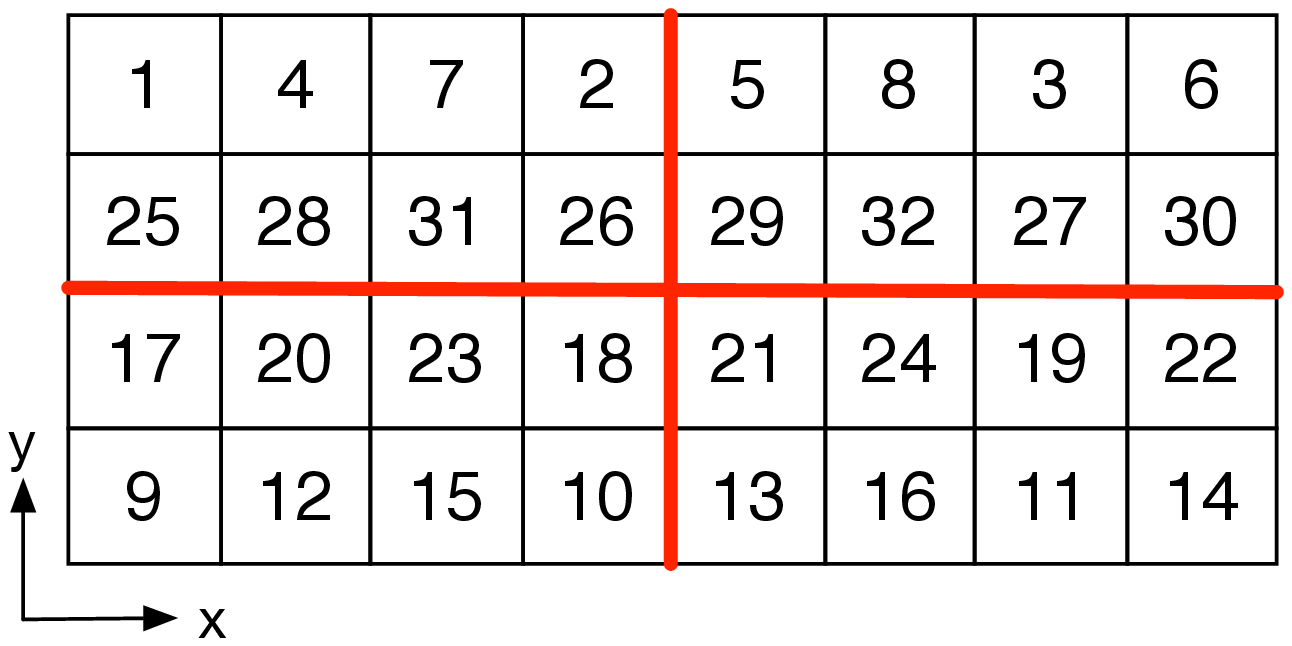} }}%
    \subfloat[]{{\includegraphics[scale=0.31]{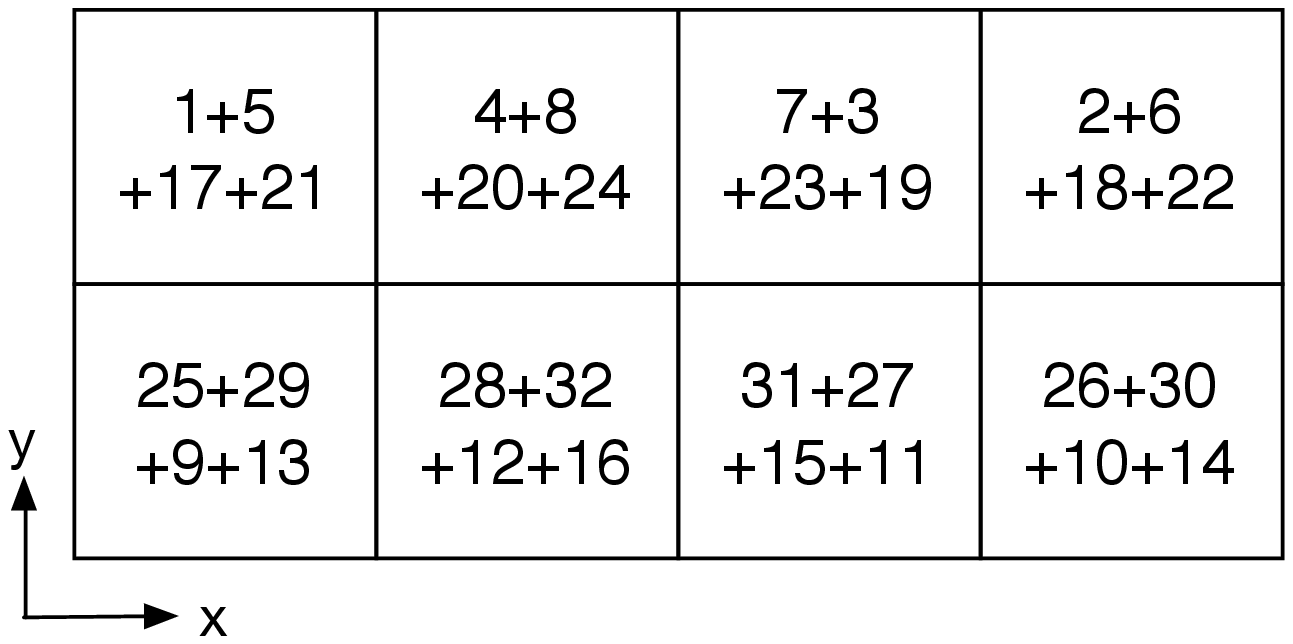} }}%
    \caption{\textbf{Permutation and Aliasing in 2D.} (a) Original 2D data forms a $4 \times 8$ matrix. (b) Permutation in $x-$dimension, $\sigma_x = 3, \tau_x = 0$. (c) Permutation in $y-$dimension, $\sigma_y = 3, \tau_y = 0$. After permutation, data is divided into four $2\times4$ sub-matrices. (d) Aliasing by adding sub-matrices from (c).}
    \label{fig:Permutation}
\end{figure}

\subsubsection{Aliasing} The aliasing stage compresses  the high dimensional data into much smaller size. In 2-D, as shown in Fig. \ref{fig:Permutation}, a periodic extension of the $N_x \times N_y$ data matrix is created with period $B_x$ in the $x$ dimension and $B_y$ in the $y$ dimension, with $B_x<<N_x$ and $B_y<<N_y$, and the basic period, i.e., $B_x \times B_y$ is extracted.

\subsubsection{First-stage-detection and reverse-mapping} We carry first stage detection after taking the square of magnitude of N-D FFT on the aliased data. Each data point is then compared with a pre-determined threshold, and for those passing the thresholds, their indices are reversed map to the original space. The combination of the reverse mapped indices from each dimension provides the tentative locations of the original frequency components. 


\subsubsection{Accumulation and second-stage-detection} The accumulation stage collects the tentative frequency locations found in the reverse mapping for each iteration, and the number of occurrences for each location is calculated after running over $T$ iterations. The second stage detection finds indices in each dimension for those whose number of occurrence pass the threshold of this stage.

Based on the above discussion, we summarize the RSFT  in Algorithm \ref{alg:RSFT}. Comparing to the original SFT, we add {pre-permutation windowing} process on the input data and incorporate NP detection in both stages of detection. 

Depending on the application, the iterations can be applied on the same input data with different permutations. Alternatively, each iteration could process different segments (see the signal model in (\ref{eq:sig_model})) of input data as indicated in Algorithm \ref{alg:RSFT}, and it can effectively reduce the variance of the estimation via RSFT,  provided that the signal is stationary.

The \emph{Estimation} procedure in the RSFT is similar to the original \emph{Estimation} procedure, except that the recovered Fourier coefficients should be divided by the spectrum of the pre-permutation window, so that the distortion due to pre-permutation windowing is compensated. Moreover, using \emph{Location} is sufficient for our radar application, for it can provide the location of significant frequencies, which is directly related to the parameters to be estimated. We also set $\tau = 0$ in each permutation, since the random phase rotation does not affect the performance of a detector after taking magnitude of the signal in the intermediate stage.

\begin{algorithm}
\caption{RSFT algorithm}\label{alg:RSFT}
\textbf{Input:} complex signal $\mathbf{r}$ in any fixed dimension \\
\textbf{Output:} $\mathbf{o}$, sparse frequency locations of input signal 

\begin{algorithmic}[1]
\Procedure{RSFT}{$\mathbf{r}$} 

\State Generate a set of $\sigma$ randomly for each dimension
\State $\bar{\mathbf{a}} \gets 0$
\For {$i \gets 0$ \textbf{to} $T$} 
	\State Pre-Permutation Windowing: $\mathbf{y} \gets \mathbf{W}\mathbf{r}$
	\State Permutation: $\mathbf{p}$ $\gets$   $P_{\sigma,0} \mathbf{y}$ 
	\State Flat-windowing: $\mathbf{z}$ $\gets$ $\overline{\mathbf{W}} \mathbf{p}$  
	\State Aliasing: $\mathbf{f} \gets  \mathop{\mathrm{Aliasing}}(\mathbf{z}$) 
	\State N-D FFT:  $\hat{\mathbf{f}}$ $\gets$ $\mathop{\mathrm{NDFFT}}(\mathbf{f})$
	\State First-stage-detection: $\mathbf{c} \gets  \mathop{\mathrm{NPdet1}}(|\hat{\mathbf{f}}|^2$)     
	\State Reverse-mapping:  $\mathbf{a}_i$ $\gets$  $\mathop{\mathrm{Reverse}}(\mathbf{c})$
	\State Accumulation: $\bar{\mathbf{a}} \gets \bar{\mathbf{a}}+\mathbf{a}_i$
\EndFor
\State Second-stage-detection: $\mathbf{o} \gets \mathop{\mathrm{NPdet2}}(\bar{\mathbf{a}})$  
\State \textbf{return} $\mathbf{o}$
\EndProcedure
\end{algorithmic}

\end{algorithm}

\section{Optimal Thresholds Design in the RSFT} \label{sec:optimal}
The main challenge in implementing the RSFT is to decide the thresholds in the two stages of NP detection. The applying of NP detection in the RSFT is not a straitforward extension on the SFT, in that the two stages are inter-connected, thus need to be jointly studied. In the following, we will show that the two asymptotically optimal thresholds are jointly founded with an optimization process.  The analysis is carried in 1-D, while the generalization to high dimension is straitforward. 
\subsection{Signal Model and Problem Formulation}
We model the signal in continuous time domain as a superposition of  $K$  sinusoids as well as additive white noise. We then sample the signal uniformly both in I and Q channels with a sampling frequency above the Nyqvist rate. Assume the total sampling time is divided into $T$ consecutive equal length segments, each of which contains $N$ samples, and $K<<N$ (i.e., the signal is sparse in frequency domain). 
Then for each time segment, i.e., $s\in[T]$, we have
\begin{equation} \label{eq:sig_model}
\mathbf{r}_s =  \sum_{i \in [K]} b_{i,s} \mathbf{v}(\omega_i) + \mathbf{n}_s,
\end{equation}
where $\mathbf{v}(\omega_i)$ denotes for the $i_{th}$ complex sinusoid, with $\omega_i \in[0,2 \pi)$ as its frequency, i.e.,
\begin{equation} \label{eq:sinusoid}
\mathbf{v}(\omega_i) = [1\quad e^{j\omega_i}\; \cdots\; e^{j(N-1)\omega_i}]^T.
\end{equation}
We further assume that $\omega_i$  is unknown deterministic quantity and is constant during the whole process, while the complex amplitude of the sinusoid, i.e., $b_{i,s}$ takes random value for each segment. More specifically, we model $b_{i,s}$ as a circularly symmetric Gaussian process with the distribution $b_{i,s} \sim \mathcal{CN}(0, \sigma_{bi}^2)$. Likewise, the noise $\mathbf{n}_s$ is distributed as $\mathbf{n}_s \sim \mathcal{CN}(\mathbf{0}, \sigma_{n}^2 \mathbf{I})$, where $\mathbf{0}$ is $N$-dimensional zero vector, and $\mathbf{I} \in \mathbb{R}^{N\times N}$ is unit matrix.  We also assume each sinusoid and the noise are uncorrelated. In addition, the neighboring sinusoids are resolvable in the frequency domain, i.e., the frequency spacing of neighboring sinusoids is greater than $\eta_m \frac{2\pi}{N}$, where $\eta_m$ is the 6.0-dB bandwidth\cite{harris1978use} of a window that applies on $\mathbf{r}_s$. Note the signal model in (\ref{eq:sig_model}) is also commonly used in array signal processing literature for Uniform Linear Array (ULA) settings (See e.g., \cite{van2002optimum}), in which the time samples in each segment is replaced by spatial samples from array elements, and a sample segment is referred as a time snapshot.

We want to detect and estimate each $\omega_i$ from the input signal. From a non-parametric and data-independent perspective, this is a classic \emph{spectral analysis} and detection problem that can be solved by any FFT-based spectrum estimation method, for example, the Bartlett method followed by a NP detection procedure (see Appendix \ref{app:Bartlett}). In that case, the FFT computes the signal spectrum on $N$ frequency bins, and the detection is carried on each bin to determine whether there exists a significant frequency.  In what follows, we define the detection and estimation problem related to the design of the RSFT.

Let $SNR_i = \sigma_{bi}^2 / \sigma_{n}^2$ be the SNR of the $i_{th}$ sinusoid. Let us define the  \emph{worst case SNR}, i.e., $SNR_{min}$, as
\begin{equation}
SNR_{min} \triangleq  \mathop{\mathrm{min}}_{i \in [K]} (SNR_i).
\end{equation}
 
Let $P_d$ denote for the probability of detection for the sinusoid with $SNR_{min}$, and $P_{fa}$ the corresponding probability of false alarm on each frequency bin.

\begin{problem} \label{prob:1}
For the signal defined in (\ref{eq:sig_model}), find the optimal thresholds of the first and the second stage of detection in an asymptotic sense, such that they minimize $SNR_{min}$ for given $P_d, P_{fa}, N, B,T, K$. Also, characterize the tradeoff between computational complexity and $SNR_{min}$ as a function of various parameters.
\end{problem}

In the following, we investigate the two stages of detection separately, then summarize the solution into an optimization problem. 

\subsection{First Stage Detection}
The first stage detection is performed on each data segment.  After pre-permutation windowing, permutation and flat-windowing, the input signal can be expressed as
\begin{equation} \label{eq:z}
\mathbf{z} = \overline{\mathbf{W}} \mathbf{P}_{\sigma_s} \mathbf{W} \mathbf{r},
\end{equation}
where $\sigma_s$ is the permutation parameter for the $s_{th}$ segment, which has an uniform random distribution; $\mathbf{P}_{\sigma_s}$ is the permutation matrix, which functions as  (\ref{eq:permutation})  with $\tau=0$;  $\mathbf{W} = \diag(\mathbf{w})$, $\overline{\mathbf{W}} = \diag(\bar{\mathbf{w}})$, where $\mathbf{w}$ and $\bar{\mathbf{w}}$ are pre-permutation window and flat-window, respectively.

Regarding the design of the flat-window, we take its frequency domain main-lobe to have width $2\pi/B$, and choose its length in time domain as $N$. As indicated in \cite{Hassanieh:2012:SPA:2095116.2095209}, it is possible to use less data in flat-windowing by choosing the length of $\bar{\mathbf{w}}$ less than $N$, i.e., dropping some samples in each segment after the permutation. However, a reduced length window in the time domain would result in longer transition regions in the frequency domain and as a result the detection performance of the system would degrade, since a larger transition region would allow more noise to enter the estimation process.

The time domain aliasing can be described as
\begin{equation} \label{eq:l}
\mathbf{f} =  \sum_{i \in [L]} \left( \overline{\mathbf{W}}_i \mathbf{P}_{\sigma_s} \mathbf{W} \mathbf{r} \right) = \mathbf{V}_{\sigma_s} \mathbf{r},
\end{equation}
where $L=N/B$; $\overline{\mathbf{W}}_{i}$ is the $i_{th}$ sub-matrix  of $\overline{\mathbf{W}}$, which is comprised of the $iB_{th}$ to the $((i+1)B-1)_{th}$ rows of $\overline{\mathbf{W}}$. And $\mathbf{V}_{\sigma_s}  =   \sum_{i \in [L]}\overline{\mathbf{W}}_i  \mathbf{P}_{\sigma_s} \mathbf{W}$.

The FFT operation on the aliased data $\mathbf{f}$ can be expressed as
\begin{equation} \label{eq:f}
\hat{\mathbf{f}} = \mathbf{D} \mathbf{V}_{\sigma_s} \mathbf{r},
\end{equation}
where $\mathbf{D}\in \mathbb{C}^{B \times B}$ is the DFT matrix.   For the $k_{th}$ entry of $\hat{\mathbf{f}}$, we have
\begin{equation} \label{eq:f_k}
[\hat{\mathbf{f}}]_k = \mathbf{u}^H_k \mathbf{V}_{\sigma_s} \mathbf{r}, \; k \in [B],
\end{equation}
where  $\mathbf{u}_k$ is the $k_{th}$ column of $\mathbf{D}$, i.e., $\mathbf{u}_k = [1\quad e^{j k \Delta \omega_B}\; \cdots\; e^{jk(B-1)\Delta \omega_B}]^T$, and $\Delta \omega_B = 2\pi/B$.

Substituting  \eqref{eq:sig_model} into (\ref{eq:f_k}), and taking the $m_{th}$ sinusoid,  which we assume is the weakest sinusoid with its SNR equals to $SNR_{min}$, out of the summation
\begin{equation} \label{eq:f_km}
\begin{split}
[\hat{\mathbf{f}}]_k &= b_m \mathbf{u}^H_k \mathbf{V}_{\sigma_s} \mathbf{v} (\omega_m) \\
&+ \sum_{j \in [K]\setminus m} \left( b_j \mathbf{u}^H_k \mathbf{V}_{\sigma_s} \mathbf{v} (\omega_j) \right)\\
&+ \mathbf{u}^H_k \mathbf{V}_{\sigma_s} \mathbf{n}.
\end{split}
\end{equation} Since $[\hat{\mathbf{f}}]_k$ is a linear combination of $b_i, [\mathbf{n}]_j, i\in[K], j\in[N]$, it holds that
 \begin{equation} \label{eq:dist_f_k}
[\hat{\mathbf{f}}]_k \sim \mathcal{CN} (0, \sigma_{fk}^2), 
\end{equation}
where 
\begin{equation} \label{eq:var_fk}
\begin{split}
\sigma_{fk}^2 &= \sigma_{bm}^2 \alpha(k, \sigma_s, \omega_m) \\
&+ \sum_{j \in [K]\setminus m} \left(\sigma_{bj}^2 \alpha(k, \sigma_s, \omega_j)\right) + \sigma_n^2 \beta(\sigma_s) ,
\end{split}
\end{equation}
and
\begin{equation} \label{eq:alpha_beta}
\begin{split}
&\alpha(k, \sigma_s, \omega) = |\mathbf{u}^H_k \mathbf{V}_{\sigma_s} \mathbf{v} (\omega)|^2 \\
&\beta(\sigma_s) =  \|\overline{\mathbf{W}} \mathbf{P}_{\sigma_s} \mathbf{w} \|^2 .
\end{split}
\end{equation} 
It is easy to see that $\sigma_{fk}^2$ is summation of weighted variance from each signal and noise component. 

We now investigate the $p_{th}$ bin, where $\omega_m$ is mapped to, and we have the following claim. 

\begin{claim}
For a complex sinusoid signal, i.e., $\mathbf{v}(\omega)$, after pre-permutation windowing, permutation with $\sigma_s$, flat windowing,  aliasing and FFT, the highest amplitude of signal spectrum appears in $[B]$ at location
\begin{equation} \label{eq:lemma13}
p(\omega, \sigma_s) = \lfloor {\frac{B}{N} [\sigma_s \lfloor \frac{\omega}{\Delta \omega_N}  \rfloor]_N} \rfloor .
\end{equation}
where $\Delta \omega_N = 2\pi/N$.
\end{claim}

\begin{proof}
If we were applying DFT to $\mathbf{v}(\omega)$, the highest amplitude of the spectrum would appear on the grid point closest to $\omega$, i.e. $\lfloor \frac{\omega}{\Delta \omega_N}  \rfloor$. The pre-permutation windowing will not change the position of the highest peak, provided the window is symmetric. Then after permutation, the peak location dilates by $\sigma_s$ modularly, and becomes $[\sigma_s \lfloor \frac{\omega}{\Delta \omega_N}  \rfloor]_N$. Finally, after flat-windowing and aliasing, the signal is ideally downsampled in the frequency domain, and the data length changes from $N$ to $B$. Then the $B$-point DFT exhibits the highest peak on grid point $p$ as desired. A visualization of this process is shown in Fig. \ref{fig:Win_Permutation}.
\end{proof}

\begin{figure}[!t]
	\begin{center}
	\includegraphics[scale=0.82]{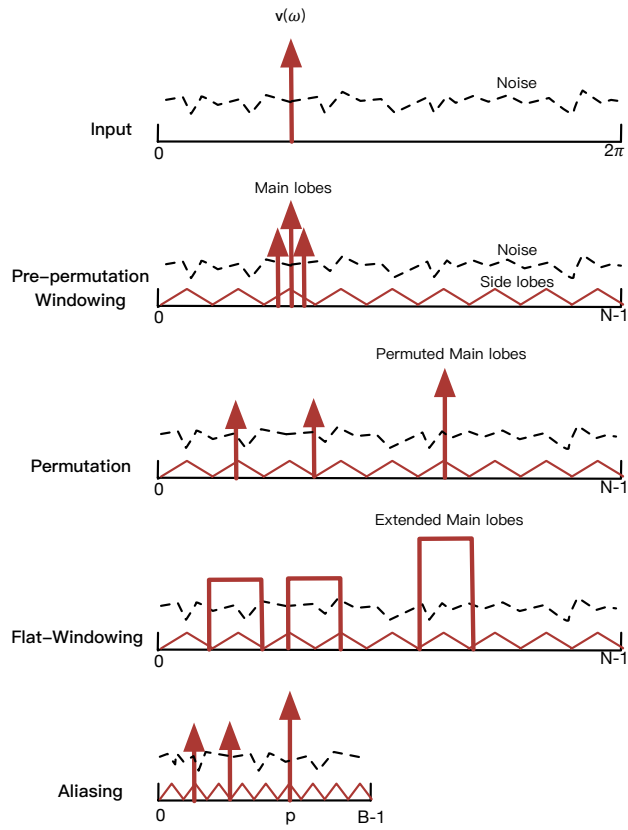} 
	\caption{\textbf{Windowing, permutation and aliasing.} The frequency representation of the signal from pre-permutation windowing to aliasing is presented. Only one significant frequency is shown for conciseness.}
	\label{fig:Win_Permutation}
	\end{center}
\end{figure} 

On assuming that only  $\omega_m$  maps to bin $p$, and that the side-lobes (leakage)  are  far below the noise level (owning to the two stages of windowing, which attenuate the leakage down to a desired level), the effect of leakage from other sinusoids can be ignored. Then we can approximate  the variance of $[\hat{\mathbf{f}}]_p$ as
\begin{equation} \label{eq:var_fp}
\begin{split}
\sigma_{fp}^2 &\approx \sigma_{bm}^2 \alpha(p, \sigma_s, \omega_m) + \sigma_n^2 \beta(\sigma_s) .
\end{split}
\end{equation}
In case that multiple frequencies are mapped to the same bin (collision), (\ref{eq:var_fp}) gives a underestimate of the variance. The probability of a collision occurring reduces as $K<<B$. 

The bin $u \in [B]$, for which no significant frequency is mapped to, contains only noise, and the corresponding variance for $[\hat{\mathbf{f}}]_u$  is
\begin{equation} \label{eq:var_fu}
\begin{split}
\sigma_{fu}^2 &\approx  \sigma_n^2 \beta(\sigma_s) .
\end{split}
\end{equation}

Hence, the hypothesis test for the first-stage-detection on $[\hat{\mathbf{f}}]_j, j\in[B]$  is formulated as
\begin{itemize}
\item $\mathcal{H}0$: no significant frequency is mapped to it. 
\item $\mathcal{H}1$: at least one significant frequency is mapped to it, with worst case SNR equals to $SNR_{min}$. 
\end{itemize}

The log likelihood ratio test (LLRT) is
\begin{equation} \label{eq:LRT1}
\log {P_{f_j|H1}(x) \over P_{f_j|H0}(x)} \mathop{\gtrless}_{\mathcal{H}0}^{\mathcal{H}1} \gamma'.
\end{equation}
where $P_{f_j|H1}(x)$ and $P_{f_j|H0}(x)$ are the probability density function (PDF) of $[\hat{\mathbf{f}}]_j$ under $\mathcal{H}1$ and $\mathcal{H}0$ respectively, and $\gamma'$ is a threshold.


Substituting the PDF of $[\hat{\mathbf{f}}]_j$ under both hypothesis into (\ref{eq:LRT1}), and after some manipulations we get
\begin{equation}
 |[\hat{\mathbf{f}}]_j|^2 \mathop{\gtrless}_{\mathcal{H}0}^{\mathcal{H}1} {{\gamma' } - { {\log {\sigma^2_{f_u} \over \sigma^2_{f_p}}}}\over{ {1 \over \sigma^2_{f_u}} - {1 \over \sigma^2_{f_p}}} }.
\end{equation}

Hence, $|[\hat{\mathbf{f}}]_j|^2$ is a sufficient statistics for first stage detection. Since $[\hat{\mathbf{f}}]_j$ has circularly symmetric Gaussian distribution, $|[\hat{\mathbf{f}}]_j|^2$ is exponentially distributed with cumulative distribution function (CDF) 
\begin{equation} \label{eq:exp}
F_{|[\hat{\mathbf{f}}]_j|^2}(x, \zeta^2) =\begin{cases}
               1 - e^{-{x \over \zeta^2}}, \; x \ge 0 \\
               0, \; x < 0 \;,
            \end{cases}
\end{equation}
where $\zeta^2$ equals to $\sigma^2_{f_u}$ under ${\mathcal{H}0}$ and  $\sigma^2_{f_p}$ under  ${\mathcal{H}1}$.

Based on (\ref{eq:exp}), in the first stage of detection, the false alarm rate on each of $B$ bins and the probability of detection of the weakest sinusoid  can be derived to be equal to
\begin{equation} \label{eq:tpd}
 \begin{split}
& \tilde{P}_{fa}(\sigma_s)  = e^{-{\gamma \over \sigma^2_n \beta(\sigma_s)}}, \\
& \tilde{P}_d(\omega_m,\sigma_s) = \tilde{P}_{fa}^{\beta(\sigma_s) \over \alpha(p, \omega_m, \sigma_s) SNR_{min} + \beta(\sigma_s)},
\end{split}
\end{equation}
where $\gamma$ is the detection threshold.  Both $\tilde{P}_{fa}$ and $\tilde{P}_{d}$  depend on the permutation $\sigma_s$. Taking expectation with respect to $\sigma_s$, we have
\begin{equation} \label{eq:bpd}
 \begin{split}
& \bar{P}_{fa} = e^{-{\gamma \over \sigma^2_n \bar{\beta}}}, \\
& \bar{P}_d(\omega_m) = \bar{P}_{fa}^{ \bar{\beta} \over  \bar{\alpha}(p,\omega_m) SNR_{min} +  \bar{\beta}},
\end{split}
\end{equation}
where $\bar{P}_d(\omega_m) = \mathbb{E} \{\tilde{P}_{d}(\sigma_s, \omega_m)\}$, $\bar{P}_{fa} = \mathbb{E} \{\tilde{P}_{fa}(\sigma_s)\}$, $ \bar{\alpha}(p, \omega_m) = \mathbb{E} \{\tilde{\alpha}(p,\omega_m,\sigma_s) \}$, and $ \bar{\beta} = \mathbb{E} \{\tilde{\beta}(\sigma_s) \}$.

\subsection{Second Stage Detection}
Let $\mathbf{c}_{\sigma_s} \in \{0,1\}^B$ denote the output of the first-stage-detection for the $s_{th}$ segment, with permutation factor $\sigma_s$. Each entry in $\mathbf{c}_{\sigma_s}$  is a Bernoulli random variable, i.e., for $j \in [B]$,
\begin{equation} 
{[\mathbf{c}_{\sigma_s}]_j} \sim \begin{cases}
               \mathrm{Bernoulli}\left(\tilde{P}_{fa}(\sigma_s)\right),  under\; \mathcal{H}0 ,\\
               \mathrm{Bernoulli}\left(\tilde{P}_{d}(\omega_m,  \sigma_s)\right),  under\; \mathcal{H}1.
            \end{cases}
\end{equation}
Note that under $\mathcal{H}1$, we assume that $[\mathbf{c}_{\sigma_s}]_j$ corresponds to the weakest sinusoid. For the other $K-1$ co-existing sinusoids, since their SNR may be grater than $SNR_{min}$, their probability of detection may also be grater than $\tilde{P}_{d}(\omega_m,  \sigma_s)$ (see Claim \ref{cl:aiH0}).

The reverse-mapping stage hashes the  $B$-dimensional   $\mathbf{c}_{\sigma_s}$ back to the $N$-dimensional $\mathbf{a}_{\sigma_s}$. According to Definition \ref{def:rev_map}, it holds that
\begin{equation} 
[\mathbf{a}_{\sigma_s}]_i = [\mathbf{c}_{\sigma_s}]_j, \; i\in[N], j\in[B], i \in \mathcal{R}(j, \sigma_s^{-1}).
\end{equation}

After accumulation of $T$ iterations, each entry in the accumulated output is summation of $T$ Bernoulli variables with different success rate. Define $\bar{\mathbf{a}}$ as the accumulated output, then for its $i_{th}, i\in[N]$ entry, we have
\begin{equation} \label{eq:a_i}
[\bar{\mathbf{a}}]_i =\sum_{s\in[T]} [\mathbf{a}_{\sigma_s}]_i = \sum_{i \in \mathcal{R}(j, \sigma_s^{-1}), s \in [T]} [\mathbf{c}_{\sigma_s}]_j .
\end{equation}
Note that in (\ref{eq:a_i}), each term inside the sum corresponds to a different segment, i.e., $[\mathbf{c}_{\sigma_s}]_{j}$ is from the $s_{th}$ segment. Since $\sigma_s$ is drawn randomly for each segment, $j$ may take different values, and relates to $i$ via a reverse-mapping. Fig. \ref{fig:Map_RevMap} gives a graphical illustration of the mapping and reverse-mapping.

Now, the hypothesis test for the second-stage-detection on $[\bar{\mathbf{a}}]_i, i\in[N]$ is formulated as
\begin{itemize}
\item $\overline{\mathcal{H}}0$: no significant frequency exists.
\item $\overline{\mathcal{H}}1$: there exists a significant frequency, whose SNR is at least $SNR_{min}$.
\end{itemize}

In the following, we investigate the statistics of $[\bar{\mathbf{a}}]_i$ under both hypothesis in an asymptotic senses. Before that however, we will take a closer look at the mapping and the reverse mapping by providing the following properties.

\begin{figure}[!t]
	\begin{center}
	\includegraphics[scale=0.82]{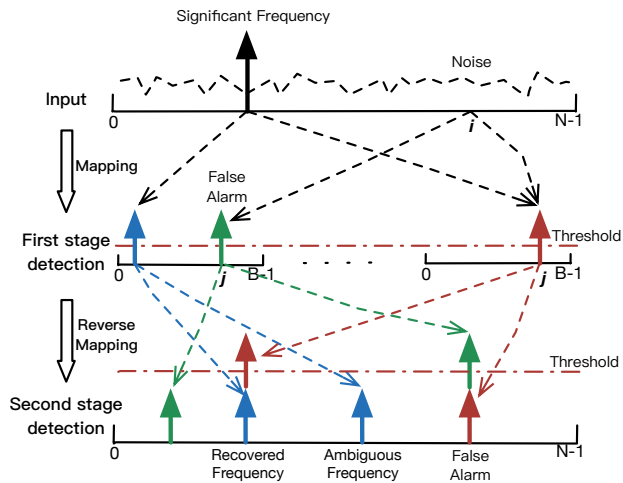} 
	\caption{\textbf{Mapping and Reverse Mapping.} A significant frequency may map to different locations in different iterations of the first stage detection, due to the different permutations. The detected frequencies, including the false alarms in the first stage,  are reverse mapped to the original dimension. The true location of the significant frequency is recovered, and the ambiguous frequencies are also generated.  The occurrence on the true location grows steadily during accumulation, provided the SNR is high enough, and thus the true location can be recovered in the second stage of detection with proper thresholding. However, false alarms may also occur in the second stage detection, due to  both ambiguous frequencies and false alarms from the first stage of detection.}
	\label{fig:Map_RevMap}
	\end{center}
\end{figure} 

\begin{property} \label{prop:mapping_reversbility}
\textbf{(Reversibility):} Let $j \in [B], i, \sigma, \sigma^{-1} \in [N]$. $\sigma$ and $\sigma^{-1}$ satisfy Eq. (\ref{eq:mod_inv}). If $j  = \mathcal{M}(i, \sigma)$, then it holds that
\begin{equation} 
i \in \mathcal{R}(j, \sigma^{-1}) .
\end{equation} 
\end{property}

\begin{property} \label{prop:mapping_exculusiveness}
\textbf{(Distinctiveness):} Let $i,j \in [B], i\neq j$. If $\sigma^{-1} \in [N]$ and satisfies Eq. (\ref{eq:mod_inv}),  then it holds that
\begin{equation} 
\mathcal{R}(i, \sigma^{-1}) \cap \mathcal{R}(j, \sigma^{-1}) = \emptyset.
\end{equation} 
\end{property}

The proofs of these properties are provided in Appendix \ref{app:map}. The two properties simply reveal the following facts: a mapped location can be recovered by reverse mapping (with ambiguities). Also, when applying reverse mapping to two distinct locations with the same permutation parameter, the resulting locations are also distinct.

Under $\overline{\mathcal{H}}1$, assuming that  $[\bar{\mathbf{a}}]_i$ corresponds to the $m_{th}$ sinusoid, i.e., the weakest sinusoid, then each term inside the sum of (\ref{eq:a_i}) has distribution  $[\mathbf{c}_{\sigma_s}]_j \sim \mathrm{Bernoulli}\left(\tilde{P}_{d}(\omega_m,  \sigma_s)\right), s\in[T]$.  Then we present the following claim.

\begin{claim} \label{cl:aiH1}
Under $\overline{\mathcal{H}}1$, and as $T \to \infty$, 
\begin{equation} \label{eq:aiH1}
[\bar{\mathbf{a}}]_i \sim N(\mu_{a1}(\omega_m), \sigma^2_{a1}(\omega_m)),
\end{equation}
where $\mu_{a1}(\omega_m) =  T\bar{P}_d(\omega_m)$,  $\sigma^2_{a1}(\omega_m) \le T\bar{P}_d(\omega_m)(1-\bar{P}_d(\omega_m))$. 
\end{claim}

\begin{proof}
Since $0<\tilde{P}_d(\omega_m, \sigma_s)(1- \tilde{P}_d(\omega_m, \sigma_s))<1$, for $\delta >0$, the Lyapunov Condition\cite{ash2000probability} holds, i.e., 
\begin{equation}
\begin{split}
&\lim_{T \to \infty} \frac{1}{\sigma_{a1}(\omega_m)^{2+\delta}} \sum_{s \in [T]} \mathbb{E}\{|[\mathbf{c}_{\sigma_s}]_j(\omega_m)-\tilde{P}_d(\omega_m, \sigma_s)|^{2+\delta}\} \\
&\le \lim_{T \to \infty} \frac{1}{\sigma_{a1}(\omega_m)^{2+\delta}} \sum_{s \in [T]} \mathbb{E}\{|[\mathbf{c}_{\sigma_s}]_j(\omega_m)-\tilde{P}_d(\omega_m, \sigma_s)|^{2}\} \\
&= \lim_{T \to \infty}  \frac{1}{\sigma_{a1}(\omega_m)^{\delta}} = 0.
\end{split}
\end{equation}

Therefore, $[\bar{\mathbf{a}}]_i$ conforms to the Normal distribution as indicated in (\ref{eq:aiH1}).  It also holds that 
\begin{equation} \label{eq:var_a1}
\begin{split}
\sigma^2_{a1}(\omega_m) &= \sum_{s \in [T]} \tilde{P}_d(\omega_m, \sigma_s)(1- \tilde{P}_d(\omega_m, \sigma_s)) \\
&= T\bar{P}_d(\omega_m)(1-\bar{P}_d(\omega_m)) \\ 
&-  \sum_{s \in [T]} (\tilde{P}_d(\omega_m, \sigma_s)-\bar{P}_d(\omega_m))^2 ,
\end{split}
\end{equation}
from which we get that $\sigma^2_{a1}(\omega_m) \le T\bar{P}_d(\omega_m)(1-\bar{P}_d(\omega_m))$, with the equality holding when $\tilde{P}_d(\omega_m, \sigma_s) = \bar{P}_d(\omega_m)$.

\end{proof}

The distribution of $[\bar{\mathbf{a}}]_i$ under $\overline{\mathcal{H}}0$ is more complicated, and we have following claim.

\begin{claim} \label{cl:aiH0}
Under $\overline{\mathcal{H}}0$, and as $T \to \infty$, 
\begin{equation} \label{eq:aiH0}
[\bar{\mathbf{a}}]_i \sim N(\mu_{a0}(\omega_m), \sigma^2_{a0}(\omega_m)),
\end{equation}
where 
\begin{equation}
\begin{split}f
\mu_{a0}(\omega_m) &= F \eta_p \bar{P}_d(\omega_m)+(T-F)\bar{P}_{fa}, \\
\sigma^2_{a0} (\omega_m) &\le F\eta_p\bar{P}_d(\omega_m)(1-\eta_p\bar{P}_d(\omega_m))\\
&+(T-F)\bar{P}_{fa}(1-\bar{P}_{fa}),
\end{split}
\end{equation}
and $F = \frac{TK\eta_m}{B}$, where $\eta_m$ is the 6.0-dB bandwidth of the pre-permutation window $\mathbf{w}$. $\eta_p \in [1, \frac{1}{\bar{P_d}(\omega_m)}]$ is a calibration factor of the probability of detection for the other $K-1$ co-existing sinusoids.
\end{claim}

\begin{proof}
Under $\overline{\mathcal{H}}0$, each term in (\ref{eq:a_i}) may be distributed differently. To illustrate this, we consider a location $i\in[N]$ in the frequency domain of the input signal, which does not contain a significant frequency, as shown in Fig. \ref{fig:Map_RevMap}.  Let $j=\mathcal{M}(i, \sigma_s)$  be the mapping. There would be two cases for $j$: 1) $j$ does not contain a significant frequency; or 2) $j$ contains at least one significant frequency, with its SNR at least $SNR_{min}$. In the former case, ${[\mathbf{c}_{\sigma_s}]_j} \sim \mathrm{Bernoulli}\left(\tilde{P}_{fa}(\sigma_s)\right)$, i.e., ${[\mathbf{c}_{\sigma_s}]_j}$ is under $\mathcal{H}0$. For the latter case, ${[\mathbf{c}_{\sigma_s}]_j} \sim \mathrm{Bernoulli}\left(\tilde{P}_{d}(\omega_m,  \sigma_s)\right)$,  i.e., ${[\mathbf{c}_{\sigma_s}]_j}$ is under $\mathcal{H}1$. Due to the permutation being uniformly random, on the average, the number of $[\mathbf{c}]_j$ under $\mathcal{H}1$ is $F = \frac{TK\eta_m}{B}$, and the number of $[\mathbf{c}]_j$ under $\mathcal{H}0$ is $T-F$. The parameter $\eta_m$ reflects the fact that sparsity is affected by the pre-permutation windowing. Since we assume that $\mathbf{v}(\omega_m)$ has the minimum SNR, i.e., $SNR_{min}$, other sinusoids with higher SNR will have larger $\bar{P}_d$. Hence we multiply $\bar{P}_d(\omega_m)$ with $\eta_p$ to calibrate the successful rate of $[\mathbf{c}_{\sigma_s}]_j$ under $\mathcal{H}1$. If all the sinusoids's SNR were equal to $SNR_{min}$, then $\eta_p = 1$; on the other hand, if the co-existing sinusoids' SNR were sufficient high so that their $\bar{P}_d$ approaches to $1$, then $\eta_p = \frac{1}{\bar{P_d}(\omega_m)}$. Finally, the results follows immediately by applying Lyapunov CLT. 
\end{proof}

\begin{remark} \label{rm:1}
From Claim \ref{cl:aiH1} and \ref{cl:aiH0}, we notice that for the second stage detection, the LLRT is obtained based on two Normal distributions. The test statistic under $\overline{\mathcal{H}}1$ is ``stable'', for it only depends on  $\bar{P}_d(\omega_m)$. However, under $\overline{\mathcal{H}}0$, the distribution depends on the number of co-existing sinusoids, as well as on each sinusoid's SNR. The larger $K$ and higher SNR will ``push'' the distribution under $\overline{\mathcal{H}}0$ closer to the distribution under $\overline{\mathcal{H}}1$, hence degrades the detection performance.  In order to compensate for this, a larger $SNR_{min}$ is required.
\end{remark}

A natural extension of Remark \ref{rm:1} is Remark \ref{re:fail}, which gives the condition under which the RSFT will reach its limit.
\begin{remark} \label{re:fail}
Assuming that $P_d \ge P_{fa}$, the RSFT will fail if $K\eta_m \ge B$ no matter how large the $SNR_{min}$ is. 
\end{remark}

\begin{proof}
Assuming $\eta_p =1$ and substituting $K\eta_m = B$ into $F$ yields $F=T$, which means that the distributions under both hypothes are the same, hence the two hypothesis cannot be discriminated. If $\eta_p > 1$, the assumption of $P_d \ge P_{fa}$ will be violated as $K\eta_m$ approaching $B$.
\end{proof}

Based on the above discussion, the optimal threshold design in Problem \ref{prob:1} can be solved by the following optimization problem, i.e.,
\begin{equation} \label{eq:opt}
\begin{split}
&Minimize_{\{\mu, \bar{P}_{fa}, \bar{P}_{d}\}} \quad SNR_{min}  \\
&Subject\; to \\ 
&\quad \quad \bar{P}_d(\omega_m) = \bar{P}_{fa}^{ \bar{\beta} \over  \bar{\alpha}(p, \omega_m) SNR_{min} +  \bar{\beta}}\\
&\quad \quad P_{fa} = \int_{\mu}^\infty g_{a_{0}} (u) du\\
&\quad \quad P_{d} = \int_{\mu}^\infty g_{a_{1}} (u) du\\
&\quad \quad 0 \le \bar{P}_{fa} \le 1,\;  0 \le \bar{P}_{d} \le 1 \\
&\quad \quad \mu \in [T] , 
\end{split}
\end{equation}
where $g_{a_{0}} (u), g_{a_{1}} (u)$ are the asymptotic  PDF\footnote{We take the upper bounds of the variances in both distributions. It is shown in Section \ref{sec:var_bounds} that the actual variances is close to their upper bounds.} of $[\bar{\mathbf{a}}]_i$ (which corresponds the weakest sinusoid) under $\overline{\mathcal{H}}0$ and $\overline{\mathcal{H}}1$, respectively. Since both of them are Normal distributions, with fixed threshold, i.e., $\mu$, we can solve for $\bar{P}_d(\omega_m), \bar{P}_{fa}$, and then compute the $SNR_{min}$. By enumerating $\mu \in [T]$, the minimum worst case SNR, i.e., $SNR_{min}^\ast$ can be found, and the corresponding $\mu^\ast$ is the optimal threshold for the second stage of detection. The optimal threshold for the first stage of detection, i.e., $\gamma^\ast$, can thus be calculated via (\ref{eq:bpd}).

\begin{remark} \label{rm:bounds}
In Claim \ref{cl:aiH0}, we set a parameter $\eta_p$ to calibrate the distribution of $[\bar{\mathbf{a}}]_i$ under $\overline{\mathcal{H}}0$. By setting $\eta_p$ as $1$ or $\frac{1}{\bar{P_d}(\omega_m)}$, we can get respectively the lower and upper bound of $SNR_{min}^\ast$ for the variation of SNR of other co-existing sinusoids. If $K$ is the maximum budget of signal sparsity,  the optimal thresholds found by solving (\ref{eq:opt}) provides the optimal thresholds for the worst case. If the actual signal sparsity were less than $K$, $P_{fa}$ would be lower than the expected value, while $P_d$ would be unchanged according to Remark \ref{rm:1}.
\end{remark}

By averaging over the permutation, asymptotically, $SNR_{min}^\ast$ does not depend on the permutation. However, it still depends on $\omega_m$, and we have the following claim to manifest their relationship.
\begin{claim} \label{cl:freqDep}
The dependence of $SNR_{min}^\ast$ on $\omega_m$ is due to the off-grid loss\cite{harris1978use} from off-grid frequencies. $SNR_{min}^\ast$ attains its minimum when $\omega_m$ is on-grid, i.e. $\omega_m = k \Delta \omega_N, k\in[N]$. When $\omega_m$ is in the middle of two grid points, i.e., $\omega_m = (k+\frac{1}{2}) \Delta \omega_N$, $SNR_{min}^\ast$ attains its maximum. 
\end{claim}

\begin{proof}
Assume $\mathbf{r} = \mathbf{v}(\omega_m)$. Since the pre-permutation window $\mathbf{w}$ is symmetric, if we applied DFT to the pre-permuted data, the amplitude of the spectrum would attain its maximum and minimum respectively when $\omega_m$ is on-grid or in the middle between two grid points. The subsequent permutation operation would not change the amplitude of the spectrum. Also, since the flat-window is used, the downsampling in the frequency domain, which is a result of aliasing, will not affect the amplitude either. The on-grid frequency generates highest amplitude, while the frequency in the middle of between grid points has the lowest amplitude. As a result, the two detection stages require the lowest SNR for on-grid frequencies, and the highest SNR for frequencies lying in the middle of between grid points. 
\end{proof}

\subsection{Tradeoff between Worst Case SNR and Complexity} \label{seg:complexity}
\subsubsection{Comparison to Bartlett Method}
We compare the complexity of the RSFT with the FFT-based Bartlett method (see Appendix \ref{app:Bartlett}) by counting the number of operations  in both algorithms as shown in Table (\ref{tb:fftComp}) and Table (\ref{tb:sftComp}). The RSFT has complexity equal to 
\begin{equation} \label{eq:RSFT_complexity}
\mathcal{O}\left(T(N+B+ B \log B+\frac{K\eta_m N}{B\eta_p})+N \right),
\end{equation}
while the Bartlett method has complexity equal to $\mathcal{O}\left(TN(1+\log N)+N \right)$.  Fig. \ref{fig:ComplexityComp} compares the RSFT's complexity to that of Bartlett's  for various $B$ and $K$. One can see that the RSFT enabled savings are remarkable when $B$ is chosen properly. Specifically, from Fig. \ref{fig:ComplexityComp}  one can see, the lowest complexity for $K$ equals to $5, 50, 100$ is achieved when $B$ equals to $32, 64, 128$, respectively. Note that the core operation in RSFT is still FFT-based, but on a reduced dimension space.  By leveraging the existing high performance FFT libraries such as FFTW \cite{frigo1998fftw}, the implementation of the RSFT algorithm could be further improved.  

\begin{remark}
The complexity of RSFT is linearly depend on $N, T, K, 1/\eta_p$ and $\eta_m$, hence it is beneficial to choose a pre-permutation window with a small $\eta_m$, provided the attenuation of the side-lobes is sufficient. We can also choose the optimal $B$ from (\ref{eq:RSFT_complexity}) to minimize the computation. However, there are  two additional constrains for $B$, one is $B$ should be a power of 2, the other is $K\eta_m \ge B$, as stated in Remark \ref{re:fail}. 
\end{remark}

\begin{table}[!t]
\caption{Computational Complexity of Bartlett Method}
\label{tb:fftComp}
\centering
\begin{tabular}{|c|c|}
 \hline
 \textbf{Procedure} & \textbf{Number of Operations} \\
 \hline
   Windowing & $TN$  \\
  \hline
  FFT & $T\frac{N}{2}\log N$ \\
   \hline
  Square & $TN$  \\
  \hline
  Detection & $N$  \\
  \hline
   Complexity & $\mathcal{O}\left(TN(1+\log N)+N \right)$ \\
  \hline
\end{tabular}
\end{table}

\begin{table}[!t]
\caption{Computational Complexity of RSFT}
\label{tb:sftComp}
\centering
\begin{tabular}{|c|c|}
 \hline
 \textbf{Procedure} & \textbf{Number of Operations} \\
 \hline
  Pre-Permutation Win & $TN$  \\
  \hline
 Permutation & $TN$ \\
  \hline
  Flat-Win & $TN$ \\
    \hline
  Aliasing & $TB(N/B-1)$   \\
 \hline
  FFT & $T\frac{B}{2}\log B$  \\
   \hline
  Square & $TB$    \\
     \hline
  First-Stage-Detection & $TB$  \\
\hline
Reverse-Mapping & $\frac{TK\eta_m N}{B\eta_p}$ \\
\hline
Second-stage-Detection & $N$ \\
  \hline
     Complexity & $\mathcal{O}\left(T(N+B+ B \log B+\frac{K\eta_m N}{B\eta_p})+N \right)$ \\
  \hline
\end{tabular}
\end{table}

\begin{figure}[!t]
	\begin{center}
	\includegraphics[scale=0.44]{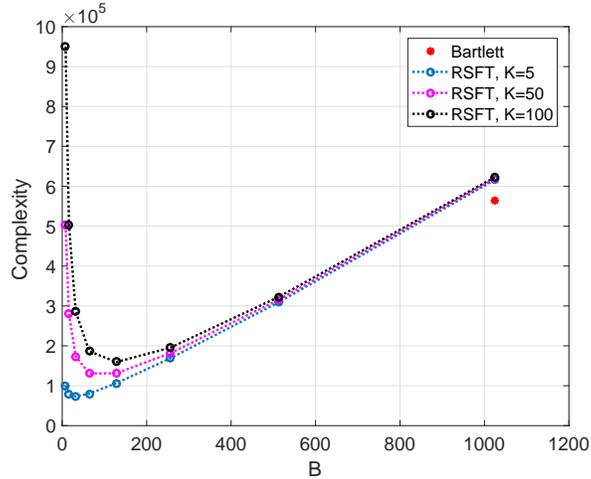} 
	\caption{\textbf{Comparison of Complexity.} $N=1024, T=50, \eta_m = 1.4,\eta_p=1, B=\{8, 16, 32, 64, 128, 256, 512, 1024\}$.}
	\label{fig:ComplexityComp}
	\end{center}
\end{figure}

\subsubsection{Worst Case SNR and Complexity Trade Off}
The reduced complexity of RSFT is achieved at a the cost of an increased $SNR_{min}$, which decreases the ability of detecting weak signals. The tradeoff between $SNR_{min}$ and complexity for various choices of parameters is shown in Fig. \ref{fig:SensCompTrad}. The performance of the Bartlett method is also shown as a reference.  From Fig. \ref{fig:SensCompTrad} we can see that  $B$ plays a central role in trading off $SNR_{min}$ and complexity.  A proper choice of $B$ can enhance the computational efficiency significantly with a reasonable increase of $SNR_{min}$. Also, since the sparsity $K$ affects both $SNR_{min}$ and complexity, a less sparse signal will worsen both. The complexity of RSFT is larger than that of the Bartlett method by setting $B=N$, due to the additional processing in the algorithm. Also, it cannot achieve the same $SNR_{min}$ as the Bartlett method does. 

\begin{figure}[!t]
	\begin{center}
	\includegraphics[scale=0.44]{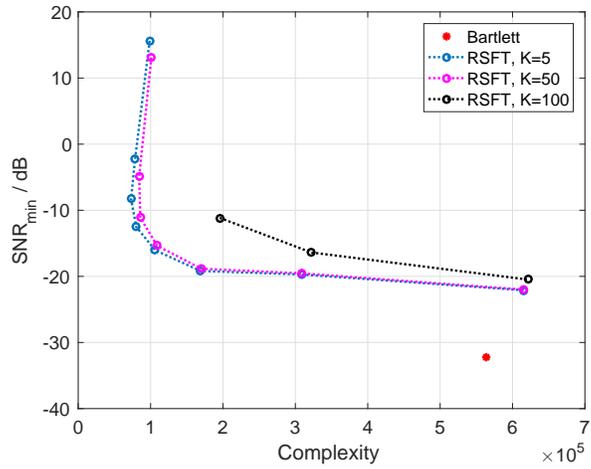} 
	\caption{\textbf{Worst Case SNR and Complexity Trade off.} $N=1024, T=50, \eta_m = 1.4, \eta_p=1, B=\{8, 16, 32, 64, 128, 256, 512, 1024\}, P_d=0.9, P_{fa} = 10^{-6}, K = \{5, 10, 100\}, \omega_m = \Delta \omega_N/2$. The red dot shows the performance of the Bartlett method, which serves as a reference.}
	\label{fig:SensCompTrad}
	\end{center}
\end{figure}

\section{Numerical Results} \label{sec:numerical}
In this section, we verify our theoretical findings via simulations. We use the following common parameters for  various settings, unless we state specifically. We take the following values: $N=1024, T=50, B=64, \eta_m=1.8, P_d =0.9, P_{fa}=10^{-6}, \omega_m = 64.5 \Delta \omega_N \approx 0.4$. We use a Dolph-Chebyshev window with $40dB$ attenuation as pre-permutation windowing. The flat-window is also based on this window, and we set its passband width as $1/B$.

\subsection{Lower and Upper Bounds of $SNR_{min}^\ast$ for Fixed Sparsity}
According to Remark \ref{rm:bounds}, we can calculate the lower bound and the upper bound of $SNR_{min}^\ast$ and their corresponding  thresholds for fixed sparsity. Fig. \ref{fig:exLowBd} and \ref{fig:exHiBd} shows the thresholding of RSFT with both bounds.  We mark the amplitude of $\omega_m$ with a magenta dot in each figure.

\begin{figure}[!t]
    \centering
    \subfloat[]{{\includegraphics[scale=0.22]{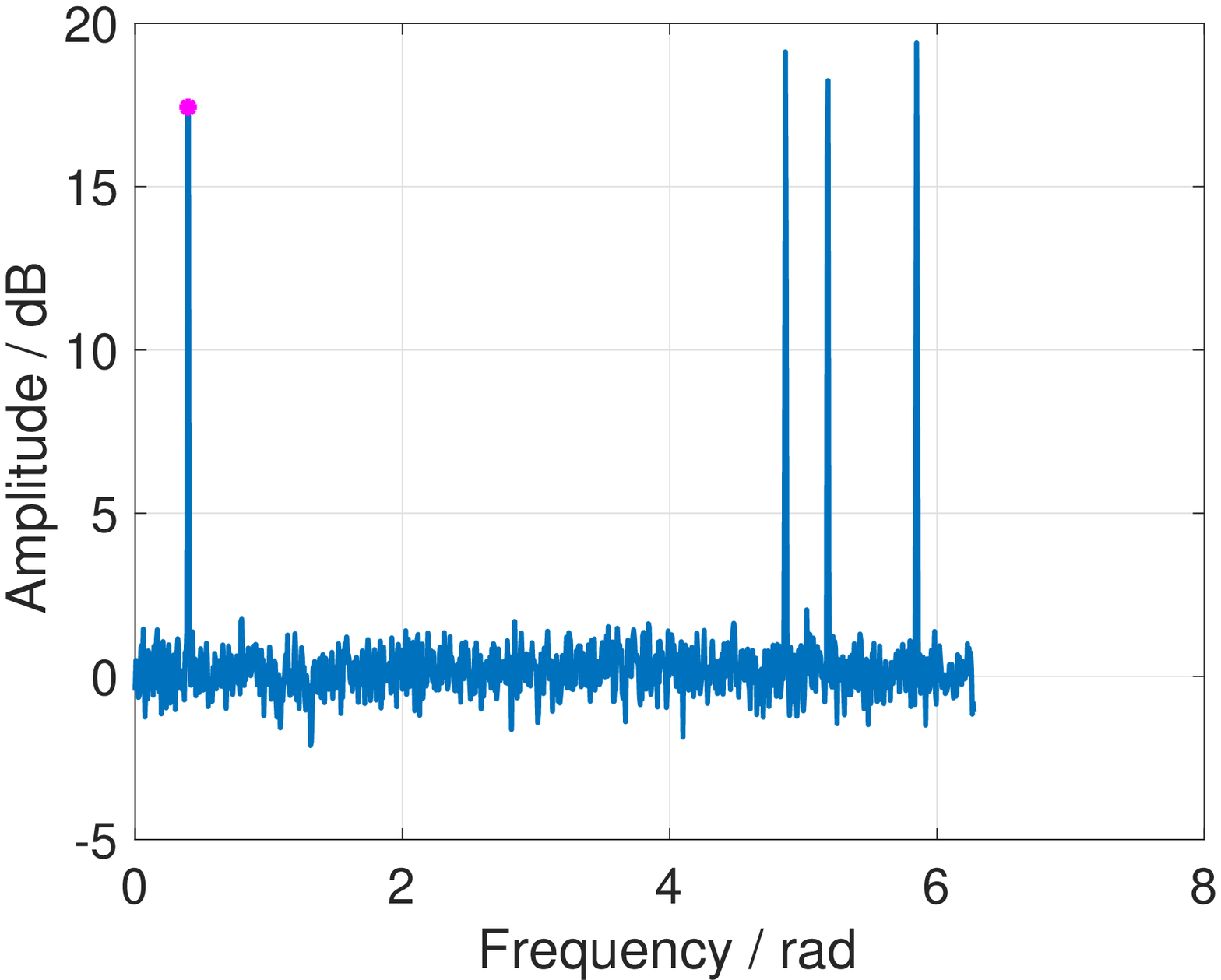} }}%
    \subfloat[]{{\includegraphics[scale=0.22]{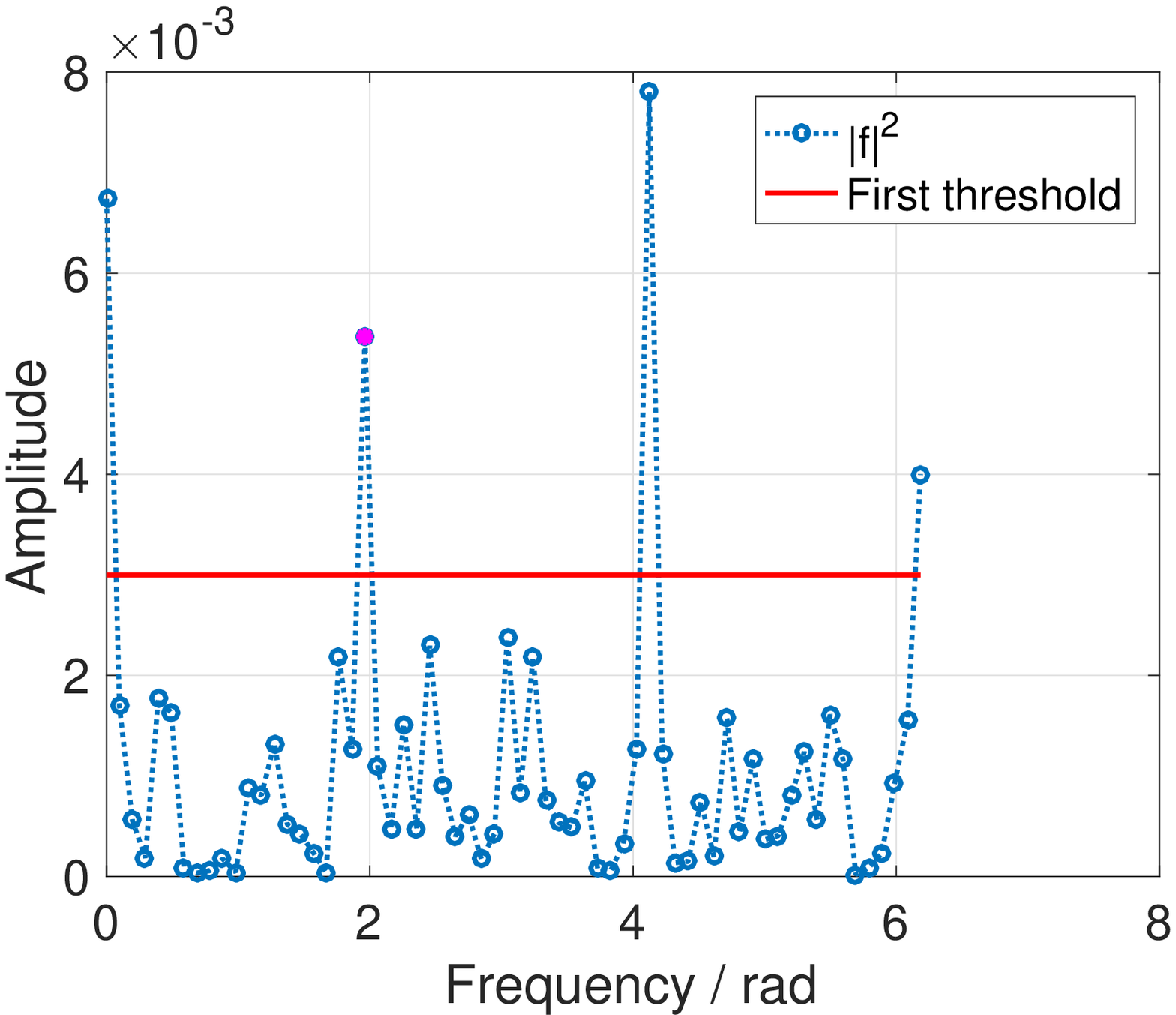} }}%
     \hfil
    \subfloat[]{{\includegraphics[scale=0.22]{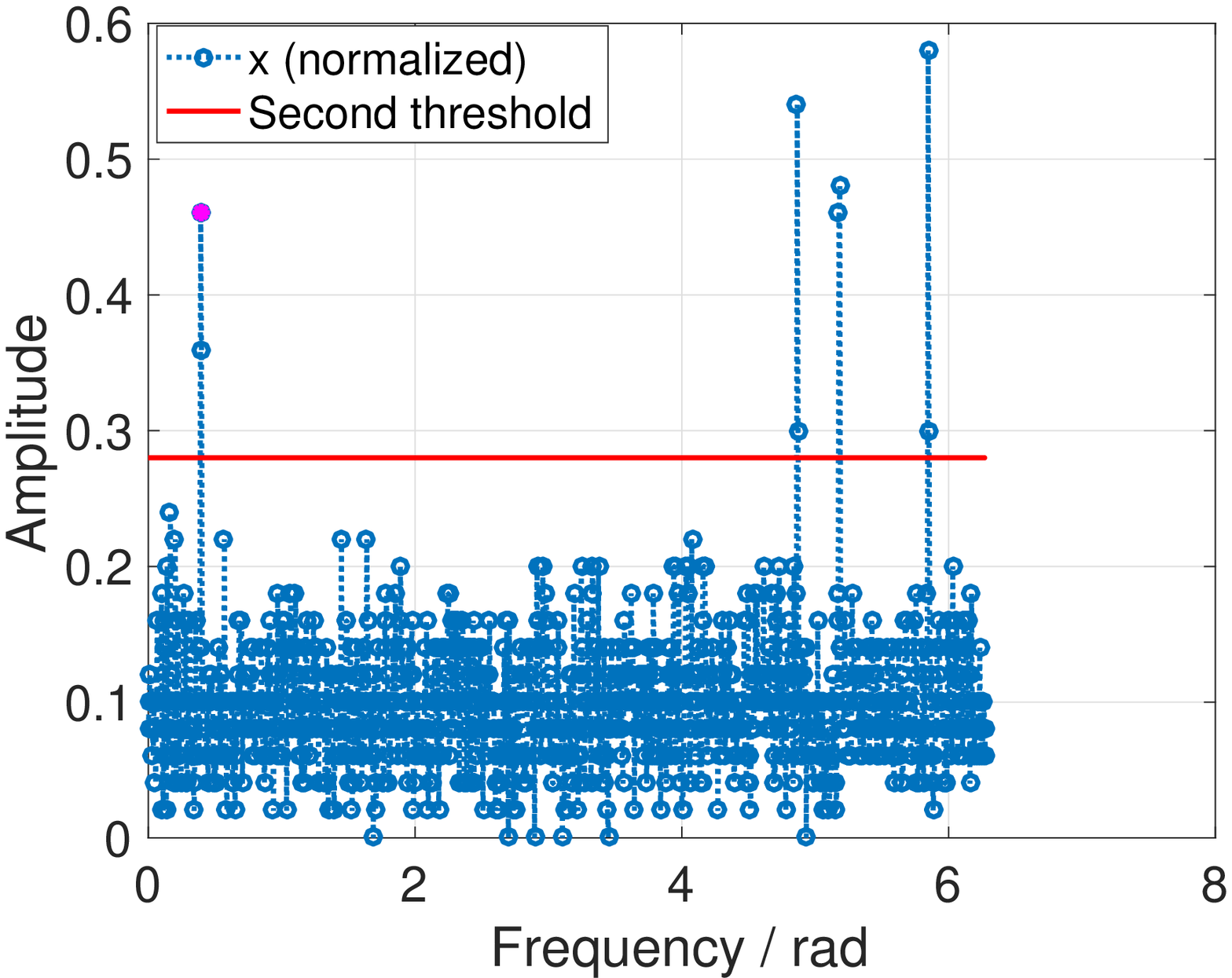} }}%
    \subfloat[]{{\includegraphics[scale=0.22]{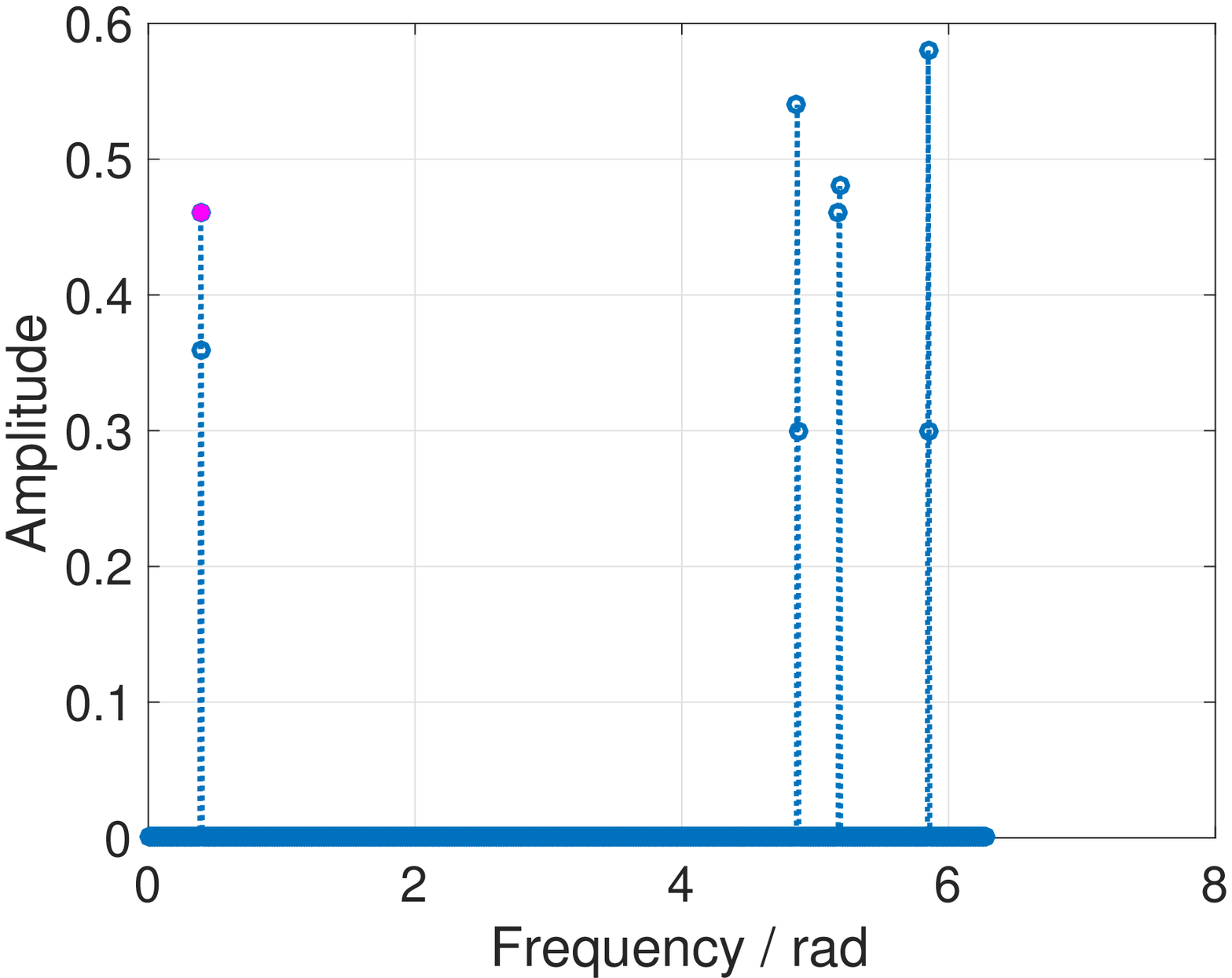} }}%
    \caption{\textbf{Lower Bound of Optimal Thresholds for Fixed Sparsity.} $K=4, SNR_{min}^\ast \approx -10.1dB, \gamma^\ast \approx 3.0\times10^{-3}, \mu^\ast/T \approx 0.3$ (a) Spectrum from Bartlett method. (b) First-stage-detection. (c) Second-stage-detection. Data and threshold is normalized by $T$. (d) Final output of RSFT.}
    \label{fig:exLowBd}
\end{figure}

\begin{figure}[!t]
    \centering
    \subfloat[]{{\includegraphics[scale=0.22]{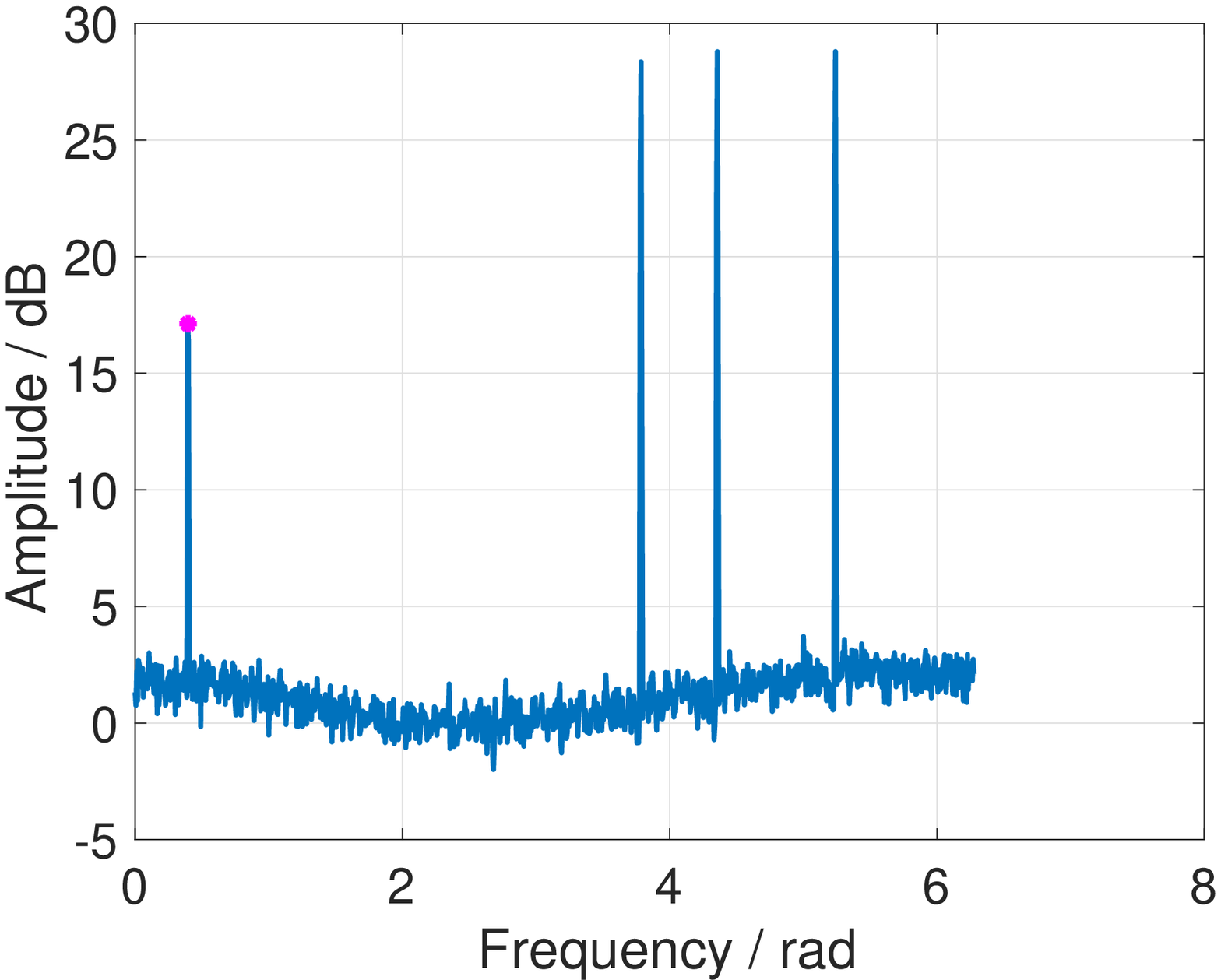} }}%
    \subfloat[]{{\includegraphics[scale=0.22]{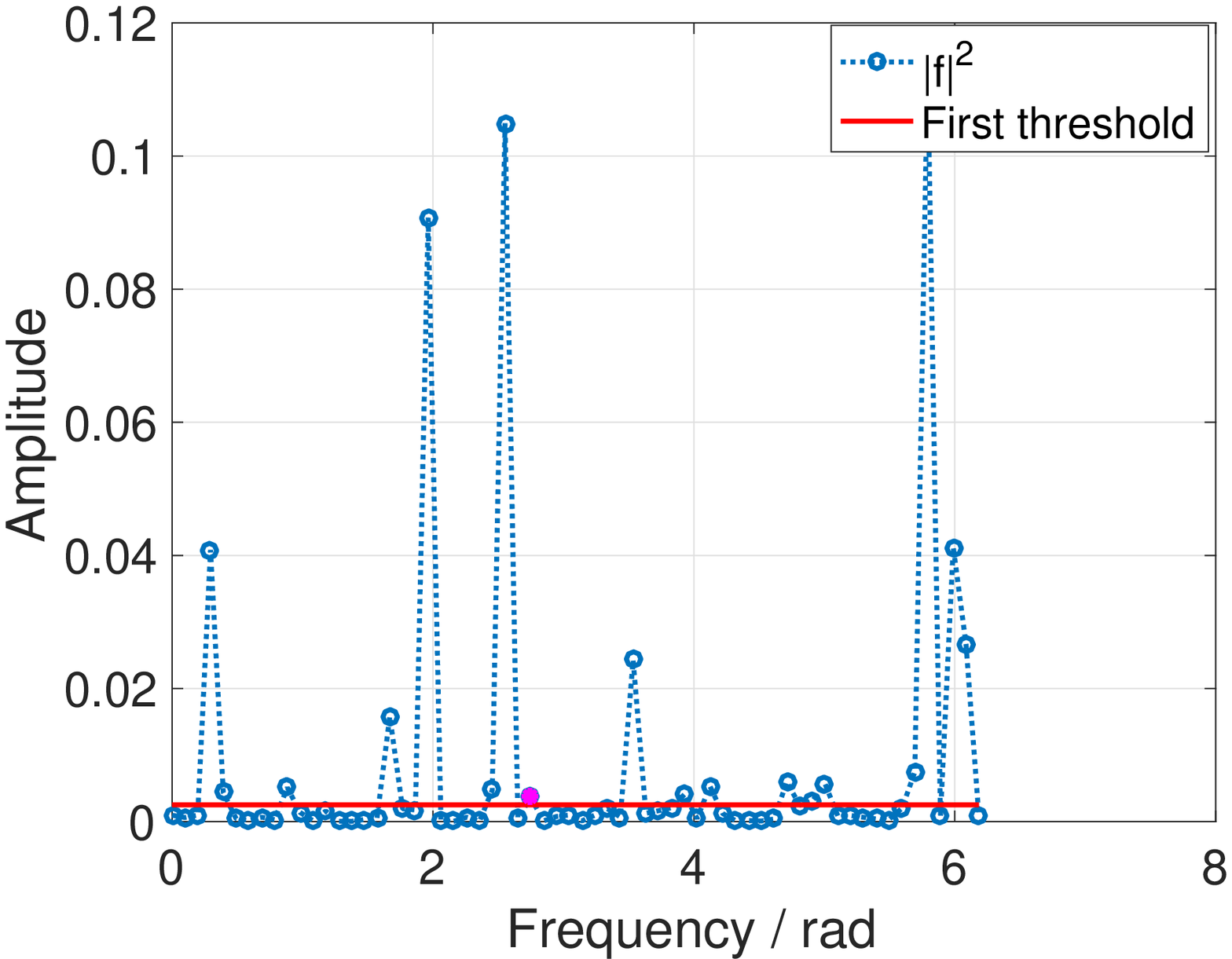} }}%
     \hfil
    \subfloat[]{{\includegraphics[scale=0.22]{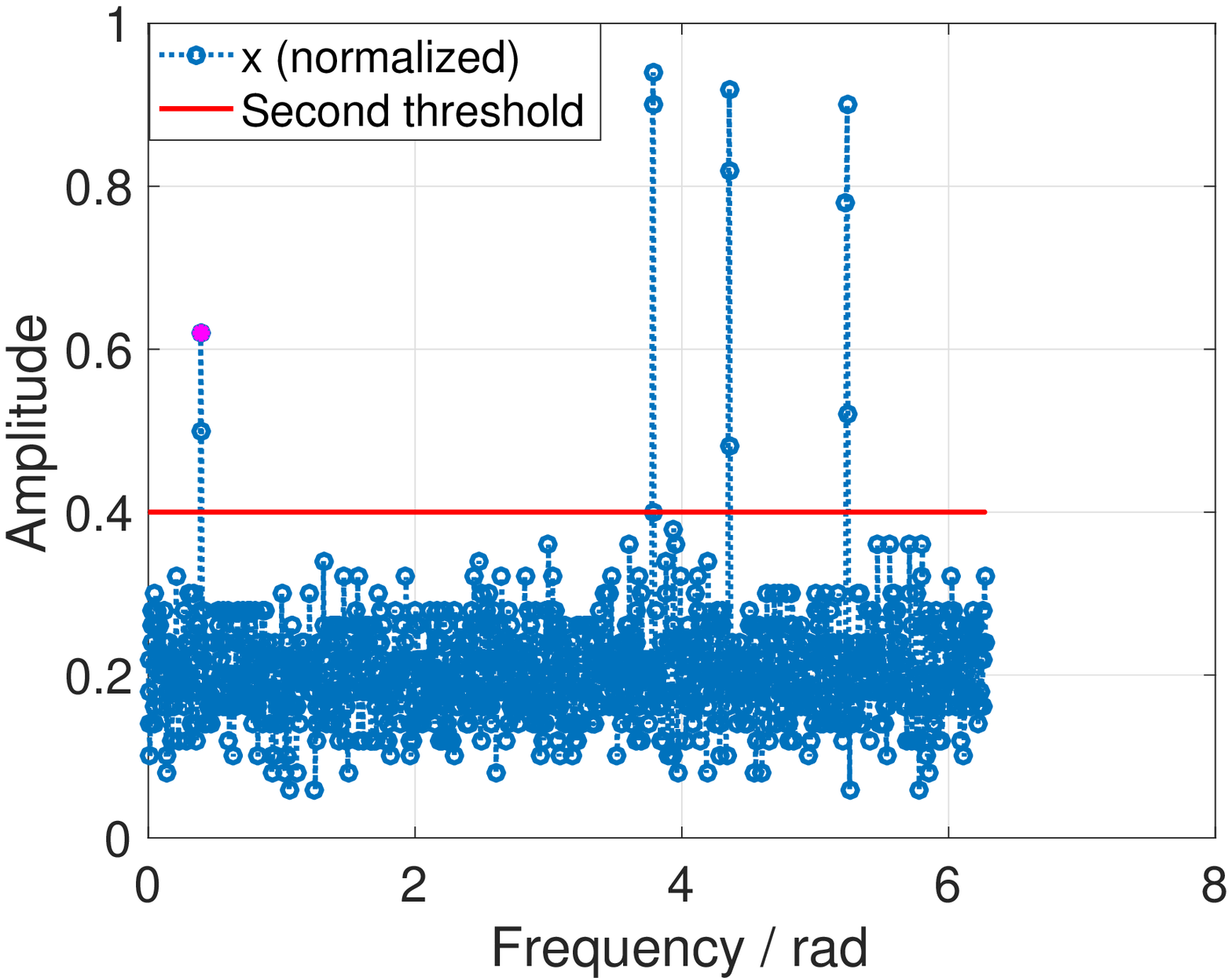} }}%
    \subfloat[]{{\includegraphics[scale=0.22]{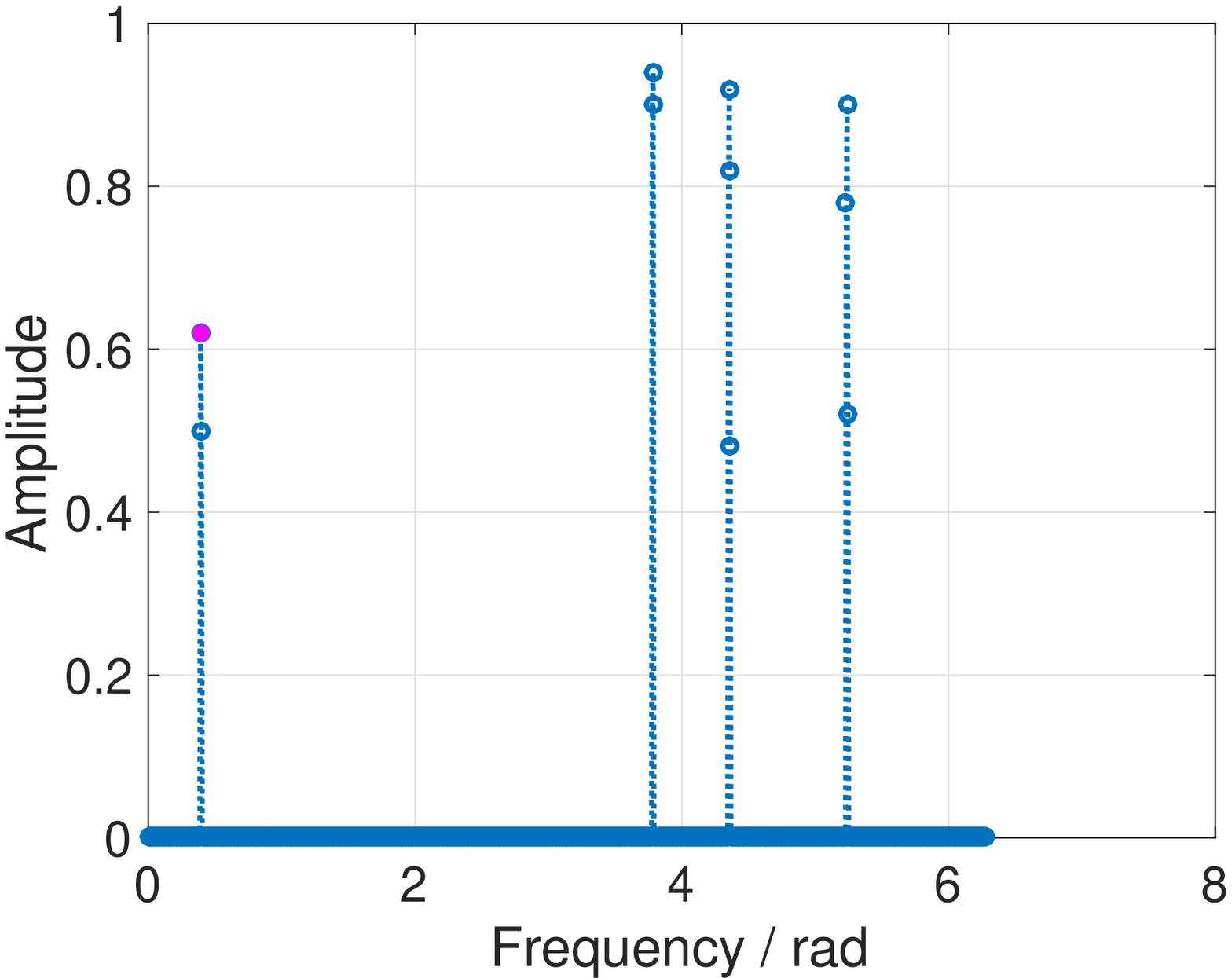} }}%
    \caption{\textbf{Upper Bound of Optimal Thresholds for Fixed Sparsity.} $K=4, SNR_{min}^\ast \approx -9.1dB,\gamma^\ast \approx 2.5\times10^{-3}, \mu^\ast/T \approx 0.4$ The co-existing sinusoids's SNR is $10dB$ higher than $SNR_{min}$ (a) Spectrum from Bartlett method. (b) First-stage-detection. (c) Second-stage-detection. Data and threshold is normalized by $T$. (d) Final output of RSFT.}
    \label{fig:exHiBd}
\end{figure}

\subsection{Unknown Signal Sparsity}
Since we do not assume that we know the exact sparsity of the signal, we will use a guess for $K$. Fig. \ref{fig:moreSparse} shows the optimal design was toward $K=10$, however, when the true sparsity is $K=3$, the system yields the same $P_d$ but better $P_{fa}$, since the noise level is much lower than expected.  

\begin{figure}[!t]
    \centering
    \subfloat[]{{\includegraphics[scale=0.22]{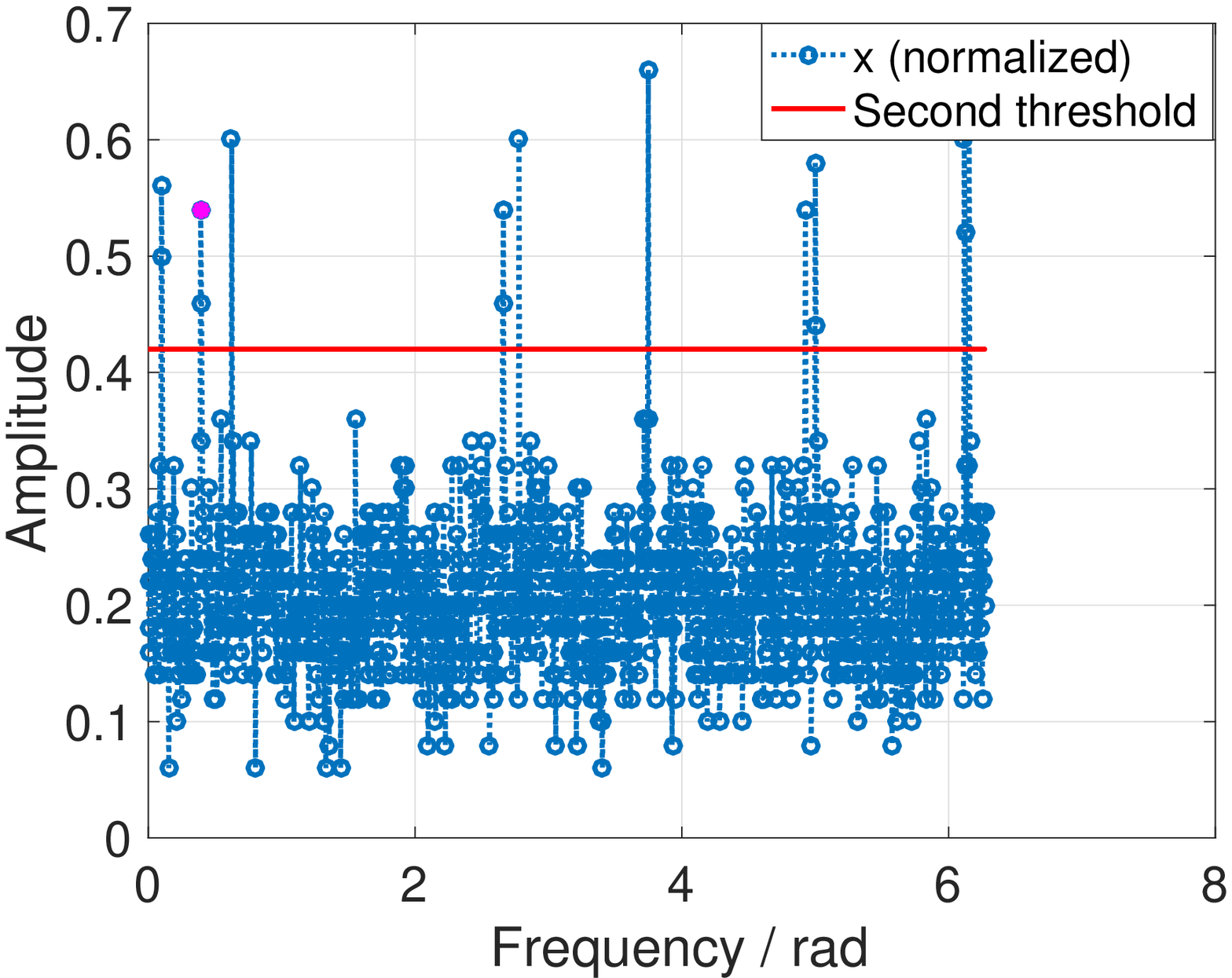} }}%
    \subfloat[]{{\includegraphics[scale=0.22]{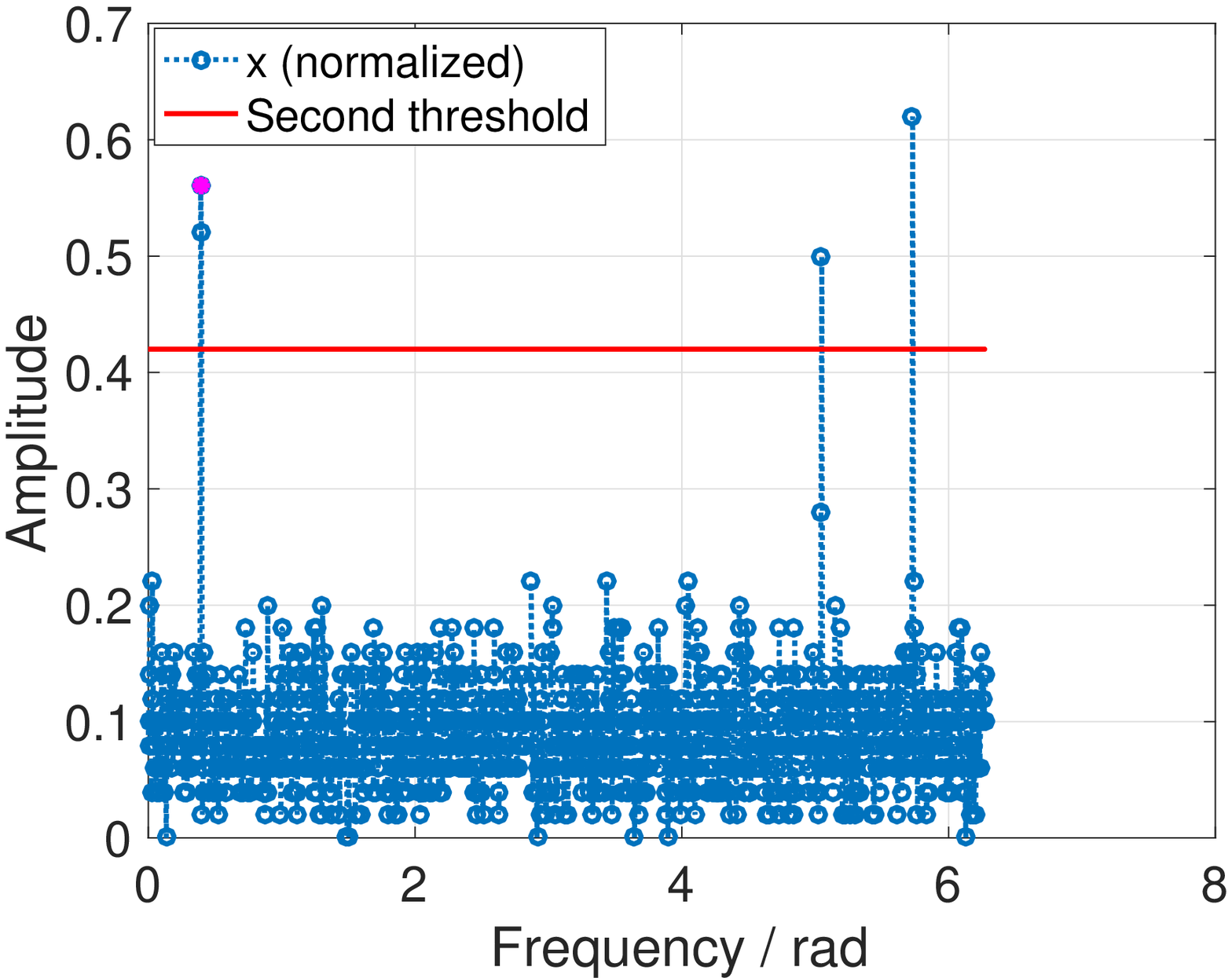} }}%
    \caption{\textbf{Unknown Sparsity.}  The threshold is optimal for $K=10$ (a) Second-stage-detection when $K=10$. (b) Second-stage-detection when $K=3$ with the same threshold. }
    \label{fig:moreSparse}
\end{figure}

\subsection{Dependency on Frequency}
Fig. \ref{fig:FreqDep} shows the dependency of $SNR_{min}^\ast$ on frequency, which verifies Claim \ref{cl:freqDep}.

\begin{figure}[!t]
	\begin{center}
	\includegraphics[scale=0.44]{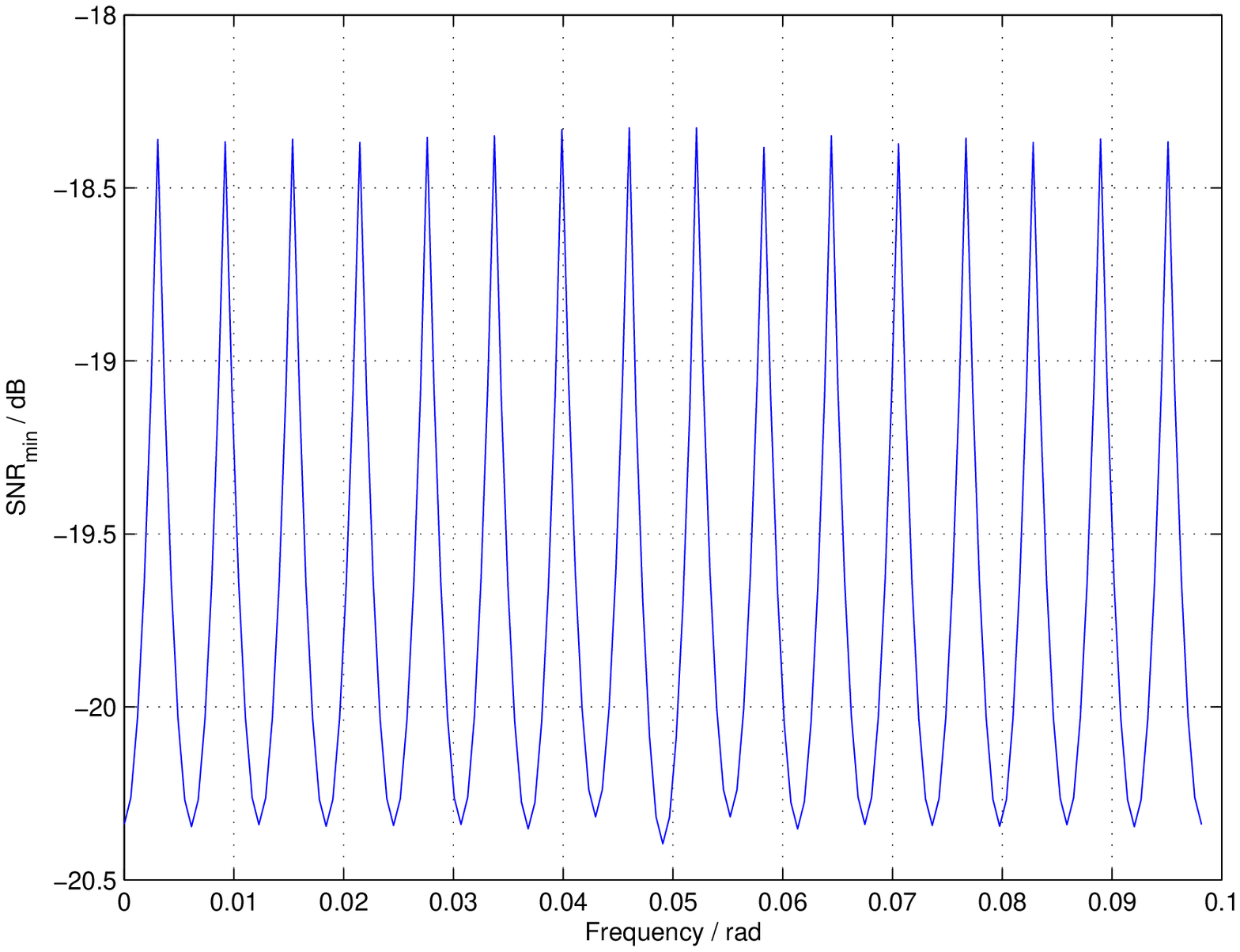} 
	\caption{\textbf{Dependence of $SNR_{min}^\ast$ on Frequency.} The fluctuation of  $SNR_{min}^\ast$ across the spectrum is mainly due to the off-grid loss of off-grid frequencies.}
	\label{fig:FreqDep}
	\end{center}
\end{figure} 

\subsection{The Receiver Operating Characteristic (ROC) Curve}
In this section, we use ROC curve to characterize the performance of RSFT with variance parameters. Fig. \ref{fig:ROC_B} shows the impact of the detection by adopting different values of $B$. A smaller $B$ lowers the detection performance, and in order to compensate it, a higher $SNR_{min}$ is required. The ROC curve for the Bartlett method is calculated by (\ref{eq:bartlettROC}) and is also shown in Fig. \ref{fig:ROC_B}.

Fig. \ref{fig:ROC_K} illustrates the relationship between detection performance and sparsity of signal. It is shown that with other parameters fixed, the sparser the signal is, the better the performance of detection is. 

In Fig. \ref{fig:ROC_CO}, we can see the impact of the SNR from the co-existing sinusoids. The higher the SNR of the co-existing sinusoids is, the worse the detection performance is. 

\begin{figure}[!t]
	\begin{center}
	\includegraphics[scale=0.44]{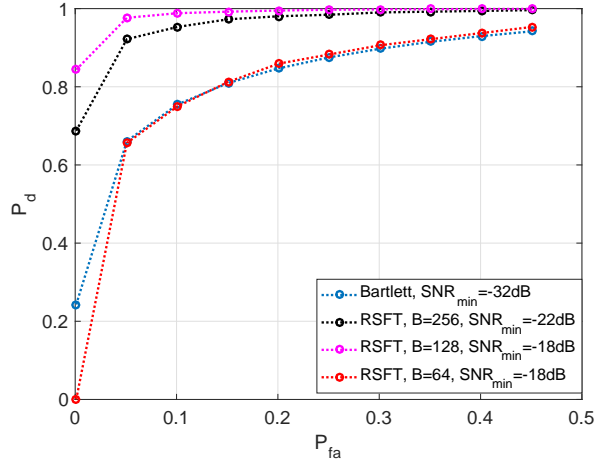} 
	\caption{\textbf{ROC Corresponding to Different Values of $B$ and $SNR_{min}$}. $\eta_p=1, K = 4$.}
	\label{fig:ROC_B}
	\end{center}
\end{figure} 

\begin{figure}[!t]
	\begin{center}
	\includegraphics[scale=0.44]{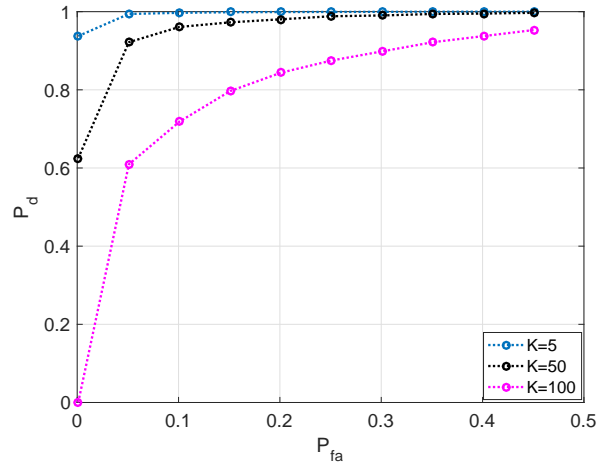} 
	\caption{\textbf{ROC Corresponding to  Different Values of $K$}. $\eta_p=1, B = 256, SNR_{min}=-20dB$.}
	\label{fig:ROC_K}
	\end{center}
\end{figure} 

\begin{figure}[!t]
	\begin{center}
	\includegraphics[scale=0.44]{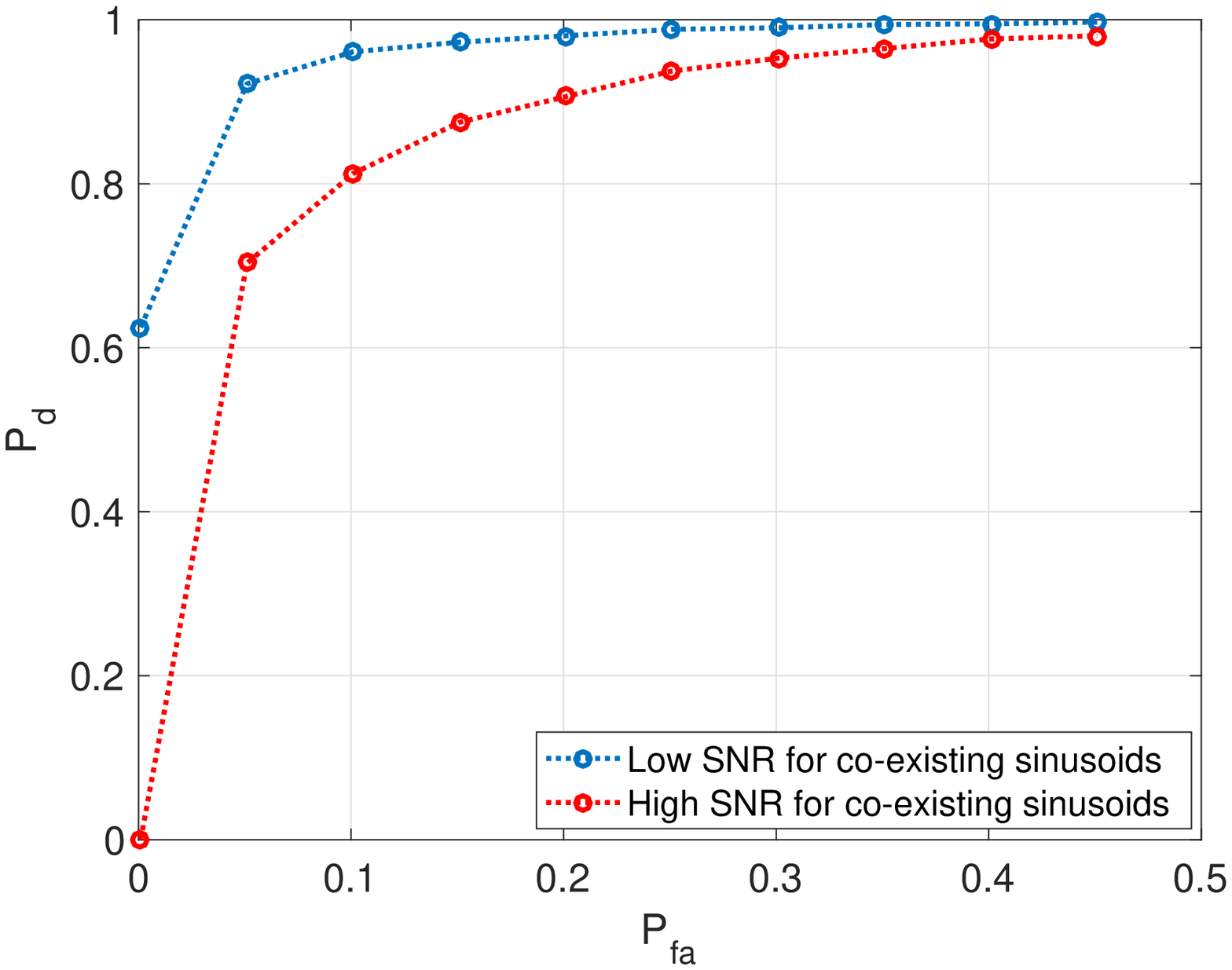} 
	\caption{\textbf{ROC Corresponding to  Different Values of SNR for Co-Existing Sinusoids}. $\eta_p=\{1, 1/\bar{P}_d(\omega_m)\}, B = 256, K=50, SNR_{min}=-20dB$.}
	\label{fig:ROC_CO}
	\end{center}
\end{figure}

\subsection{The Variance and Its Upper Bound for $[\bar{\mathbf{a}}]_i$} \label{sec:var_bounds}

In solving (\ref{eq:opt}), we take the upper bound of the variance for $[\bar{\mathbf{a}}]_i$ under both hypotheses. In this section, we show by simulation that the actual variance of $[\bar{\mathbf{a}}]_i$ is close to its upper bound. In what follows, we study $\sigma_{a1}^2(\omega_m)$; the case for $\sigma_{a0}^2(\omega_m)$ can be similarly studied. 

As shown in (\ref{eq:var_a1}), the discrepancy of $\sigma_{a1}^2(\omega_m)$ from its upper bound is due to the $\tilde{P}_d$'s dependence on $\sigma_s$, which is caused by $\beta$ and $\alpha$'s dependence on $\sigma_s$ (see (\ref{eq:tpd})). For $N$ a power of 2, a valid $\sigma_s$ can be any odd integer in $[N]$\cite{Hassanieh:2012:SPA:2095116.2095209}. Fig. \ref{fig:alpha} shows $\beta(\sigma_s)$ and $\alpha(p, \omega_m, \sigma_s)$ as functions of $\sigma_s$. The symmetry of the plot is due to the symmetry of pre-permutation window and the flat-window, as well as the modulo property of the permutation. Another observation is that most of  $\beta(\sigma_s)$ and $\alpha(p, \omega_m, \sigma_s)$ have similar values. As a result, $\tilde{P}_d(\omega_m, \sigma_s)$ has similar value for different permutations, and this is the reason for $\sigma_{a1}^2(\omega_m)$ being close to its upper bound. The Monte Carlo simulation in Fig. \ref{fig:varAppErr} shows that the approximation error, i.e., $\frac{T\bar{P}_d(\omega_m)(1-\bar{P}_d({\omega_m})) - \sigma_{a1}^2(\omega_m)}{\sigma_{a1}^2(\omega_m)}$ decreases as $T$ grows, and even for a small $T$, such as $T=10$, the error is as small as about $1.6\%$.

\begin{figure}[!t]
	\begin{center}
	\includegraphics[scale=0.44]{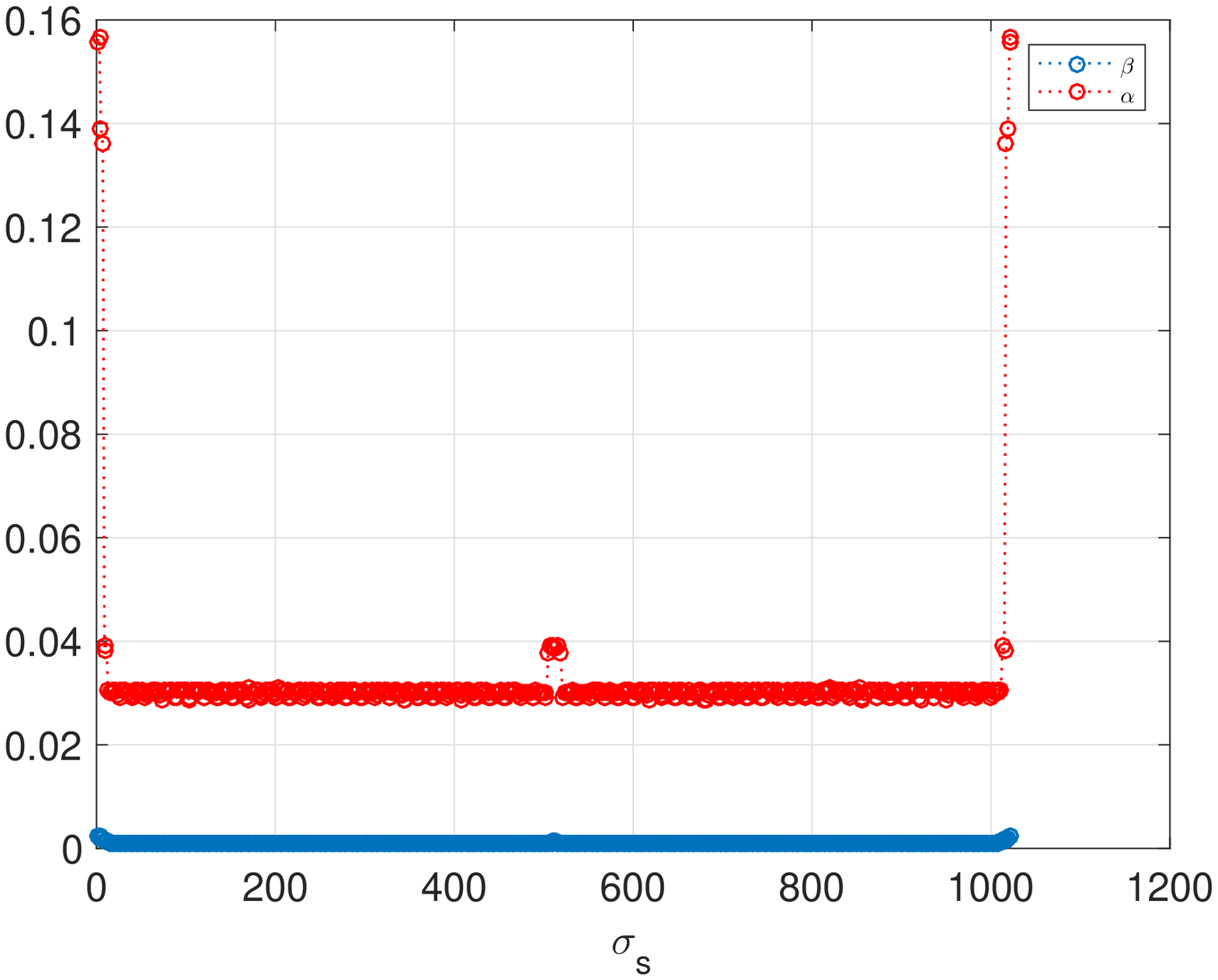} 
	\caption{\textbf{$\beta$ and $\alpha(\omega_m)$'s dependence on permutation.} $B= 64$.}
	\label{fig:alpha}
	\end{center}
\end{figure} 

\begin{figure}[!t]
	\begin{center}
	\includegraphics[scale=0.44]{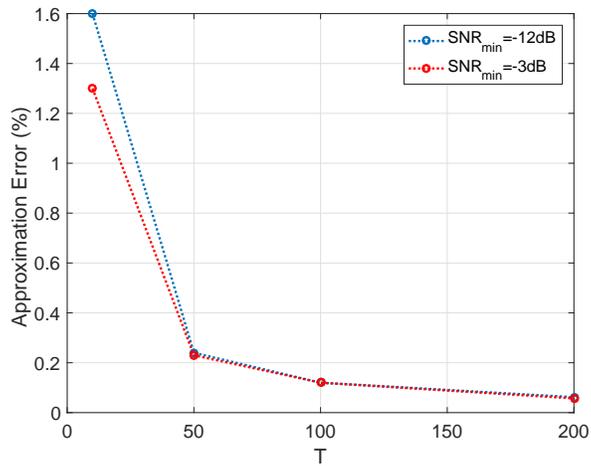} 
	\caption{\textbf{Approximation Error for the Variance.} $B = 64, T=\{10, 50, 100, 200\}$. Each value is averaged over 100 times Monte Carlo simulation.}
	\label{fig:varAppErr}
	\end{center}
\end{figure}

\section{RSFT for Ubiquitous Radar Signal Processing} \label{sec:application}
The RSFT algorithm can greatly reduce the complexity of certain high dimensional problems.   This can be signifiant in many applications, since lower complexity means faster reaction time and more economical hardware.   However, in order to apply RSFT, the signal to be processed should meet the following requirements:
\begin{itemize}
\item It should be sparse in some domain.
\item It should be sampled uniformly whether in temporal or spacial domain.
\item The SNR should be moderately high.
\end{itemize}
While many applications satisfy these requirements, in what follows, we discuss an example in Short Range Ubiquitous Radar\cite{skolnik2002systems} (SRUR) signal processing. 

\subsection{Short Range Ubiquitous Radar}
An ubiquitous radar or SIMO radar  can see targets everywhere at anytime without steering its beams as a traditional phased array radar does. In SRUR, a broad transmitting beam patten is achieved by an omnidirectional transmitter and multiple narrow beams are formed simultaneously after receiving of the reflected signal.  The beam pattens of an ubiquitous radar is shown in Fig. \ref{fig:MIMOtr} with an Uniform Linear Array (ULA) configuration. 

\begin{figure}[htp!]
	\begin{center}
	\includegraphics[scale=0.4]{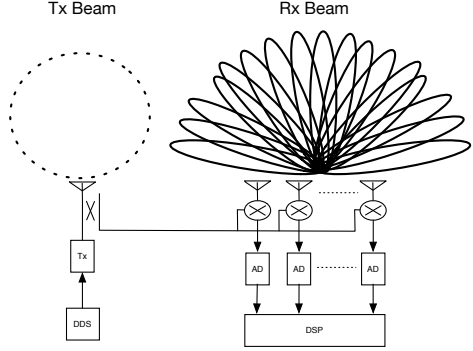} 
	\caption{\textbf{Ubiquitous Radar System Structure and Transmitting / Receiving Beam Patten.} A broad beam patten is transmitted with an omnidirectional antenna, while multiple narrow beams are formed simultaneously by the receiving array. Each receiving channel is mixed with a coupled signal from the transmitter to de-chirp the LFMCW signal, before the A/D converting.}
	\label{fig:MIMOtr}
	\end{center}
\end{figure}

An SRUR with range coverage of several kilometers could be important both in military and civilian vehicular applications. For instance, in an active protection system \cite{schade2007fast}, sensors on the protected vehicle have to detect and locate the warheads from a closely fired rocket-propelled grenade (RPG) within milliseconds. Among other sensors, SRUR's simultaneous wide angle coverage, high precision of measurement and all-weather operation make it the ideal sensor for such situation. 

\begin{figure}[!t]
	\begin{center}
	\includegraphics[scale=0.7]{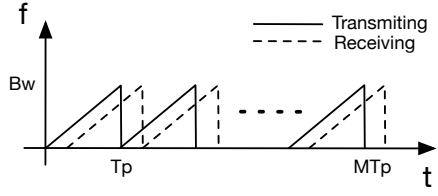} 
	\caption{\textbf{LFMCW Waveform.} The signal frequency change linearly in time with a repetition internal $T_p$. A burst contains $M$ repetitions. The range of frequency changing is the bandwidth of the system. The received signal is a delayed version of the transmitted signal.}
	\label{fig:wave}
	\end{center}
\end{figure} 

In order to achieve high range resolution and cover near range, SRUR utilizes a LFMCW waveform, as shown in Fig. \ref{fig:wave}. Mathematically, the transmitted waveform can be expressed as
\begin{equation} \label{eq:trsig}
s(t,v) = A \cos (2 \pi (f_c (t-v T_p) + \pi \rho (t-v T_p)^2),
\end{equation}
where $T_p$ is the  repetition interval (RI), $v \in [M]$ denotes the $v_{th}$ RI, $A$ is amplitude of the signal, $f_c$ is the carrier frequency and $\rho$ is the chirp rate.  Furthermore, without loss of generality, we assume that the initial phase of the signal is zero. 

Upon reception, a de-chirp process is implemented by mixing the received signal with the transmitted signal, followed by a lowpass filter. The received signal is a delayed version of the transmitted one, hence by mixing the two signals, the range information of the targets is linearly encoded in the difference of the frequencies. Hence for the $i_{th}, i\in[N]$ receiving channel, the de-chirped signal is expressed as
\begin{equation} \label{eq:r_i}
\begin{split}
r_i &= \sum_{k \in [K]} a^{[k]}(s) \cos \left(2 \pi ((f_r^{[k]}+f_d^{[k]}) (t-v T_p) + i \pi \sin \theta^{[k]} \right)\\
 &+ n(t),
\end{split}
\end{equation}
 which is a superposition of $K$ sinusoids and additive noise $n(t)$. For the $k_{th}$ sinusoid, $a^{[k]}(s), s \in [T]$ represents its amplitude, which can be modeled as a Gaussian random process. More specifically, the amplitude is assumed static within a burst, and independent between each of $T$ bursts. This assumption is consistent with the Swerling-\uppercase{I} target model\cite{skolnik1970radar}, which represents a slow fluctuation of the target RCS.  $f_r^{[k]}, f_d^{[k]}$ are the frequency components respect to target's range and velocity respectively, i.e.,
 \begin{equation} \label{eq:s_u}
f_r^{[k]}  =\frac{2\rho r_t^{[k]}}{c}, \hskip10pt  f_d^{[k]} = \frac{2v_t^{[k]}}{\lambda},
\end{equation}
where $r_t^{[k]}, v_t^{[k]}, c$ are the $k_{th}$ target's range, velocity and speed of wave propagation respectively. 
 
The DOA of the $k_{th}$ target, i.e., $\theta^{[k]}$ is defined as the angle between the line of sight (from the array center to the target) and the array normal. Assuming that the element wise spacing is $\lambda/2$, under the narrowband signal assumption, $\theta^{[k]}$ will cause an increase of phase at the neighboring array element equal to $\pi \sin \theta^{[k]}$. We omit the constant phase term in each sinusoids of (\ref{eq:r_i}), since they are irrelevant to the performance of the algorithm.
 
After AD conversion of each receiving channel, we can use the processing scheme shown in Fig. \ref{fig:MIMO_Bart} to detect the targets as well as estimate their range, velocity and DOA. More specifically, grid-based versions of $f_r^{[k]}, f_d^{[k]}, \pi \sin \theta^{[k]}$ can be calculated by applying a 3-D FFT on the windowed data cube, then, after accumulation of $T$ iterations,  the above described NP detection procedure can be performed.

\subsection{RSFT-based SRUR Signal Processing}
Although the number of samples of SRUR is reduced significantly with the analog de-chirp processing, the realtime processing with 3-D FFT is still challenging.  The RSFT algorithm is suitable for reducing the computational complexity of SRUR, since, 1) the number of targets is usually much smaller than the number of spatial resolutions cells, which implies that the signal is sparse after proper translation; 2) with a ULA and digitization of each received element, the signal is uniformly sampled both in spatial and temporal domain; and 3) the short range coverage implies that moderate high SNR is easy to achieve. 

The RSFT-based SRUR processing architecture is shown in Fig. \ref{fig:MIMO_SFT}. Compared to the conventional processing, the 3-D FFT is replaced with a 3-D RSFT, in which the aliasing procedure converts the data cube dimensions from $R \times N \times M$ to $B \times C \times D$. The 3-D FFT operated on the smaller data cube could save the computation time significantly.

\begin{figure}[!t]
	\begin{center}
	\includegraphics[scale=0.52]{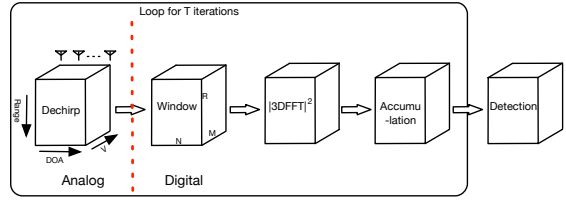} 
	\caption{\textbf{Conventional FFT-based Processing Scheme for SRUR.} }
	\label{fig:MIMO_Bart}
	\end{center}
\end{figure} 

\begin{figure}[!t]
	\begin{center}
	\includegraphics[scale=0.56]{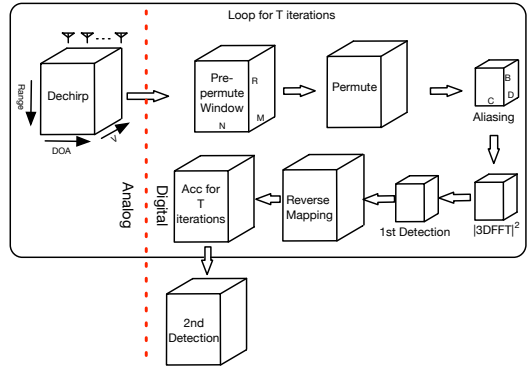} 
	\caption{\textbf{RSFT Based Processing Scheme for SRUR.} }
	\label{fig:MIMO_SFT}
	\end{center}
\end{figure} 

\subsection{Simulations}
In this section, we verify the feasibility of RSFT-based SRUR processing and compare to the SFT-based processing via simulations. The main parameters of the system are listed in Table \ref{tb:SRUR}. The design of the system can guarantee  non-ambiguous measurements of the target's range and velocity, assuming the maximum range and velocity are less than $1.5km$ and $300 m/s$, respectively.

\begin{table}[!t]
\caption{SRUR Parameters}
\label{tb:SRUR}
\centering
\begin{tabular}{|c|c|c|}
 \hline
 \textbf{Parameter} & \textbf{Symbol} & \textbf{Value} \\
 \hline
   Number of range bins& $R$& $2048$  \\
  \hline
 Number of receiving elements& $N$ & $64$  \\
  \hline
   Number of RI &$M$& $32$  \\
   \hline
   Wave length &$\lambda$& $0.03m$  \\
   \hline
    Wave propagation speed &$c$& $3\times 10^8 m/s$  \\
   \hline
       Bandwidth &$B_w$& $150MHz$  \\
   \hline
       Repetition interval &$T_p$& $5\times10^{-5}s$  \\
   \hline
          Maxima range &$R_{max}$& $1.5\times10^{3}m$  \\
   \hline
          Chirp rate & $\rho$ &$3\times10^{12}Hz/s$  \\
   \hline
          Sampling frequency (IQ) & $f_s$ &$41MHz$  \\
   \hline
\end{tabular}
\end{table}

We generate signals from $4$ targets according to (\ref{eq:r_i}). The range, velocity and DOA of targets can be arbitrarily chosen within the unambiguous space, which implies that the corresponding frequency components do not necessarily lie on the grid. The targets' parameters used in the simulation are listed in Table \ref{tb:target}. For Targets $3$ and $4$, we set the same parameters except that their DOA is $4^\circ$ apart, which is close to the theoretical angular resolution after windowing for the Bartlett beamforming. To compare the RSFT and the SFT for different scenarios, we adopt two sets of SNR for targets. Specifically, for the first set, we use the same SNR, i.e., $-10dB$ for different targets. And for the second set, we assign different SNR for different targets, which is more close to a realistic scenario. 

\begin{table}[!t]
\caption{Target Parameters}
\label{tb:target}
\centering
\begin{tabular}{|c|c|c|c|c|}
 \hline
 \textbf{Target} & \textbf{Range ($m$)} & \textbf{Velocity ($m/s$)} &  \textbf{DOA ($\circ$)} & \textbf{SNR (dB)}\\
 \hline
1& $1000$& $100$ & $30$ & $-10/0$\\
  \hline
  2& $500$& $50$ & $0$ & $-10/-10$\\
  \hline
  3& $350$& $240$ & $-16$ & $-10/-20$\\
  \hline
  4& $350$& $240$ & $-20$ & $-10/-20$\\
  \hline
\end{tabular}
\end{table}

The SFT from \cite{Hassanieh:2012:SPA:2095116.2095209} is $1$-dimensional. In order to reconstruct targets in the 3-D space, we extend the SFT to high dimension with the techniques described in Section \ref{sec:hd}. Another obstacle of applying SFT is that it needs to know the number of peaks to count in the detection stages. In our radar example, even the knowledge of exact number of targets is presented, it is still not clear how to determine the number of counting peaks due to existence of the large number of peaks from leakage.
In the experiment, we gradually increase the number of counting peaks until all the targets are recovered. For the case of the same SNR setting, all the targets are recovered after around $20$ peaks are counted. While for the second SNR setting, we need to count nearly $200$ peaks to recover the weakest targets (Targets $3$ and $4$). Fig. \ref{fig:3DFFT} and \ref{fig:3DSFT1} show the targets reconstruction results for the two settings, respectively. The former shows both SFT and RSFT methods can perfectly recover all the targets, whose SNR has the same value. From Targets $3$ and $4$ we can see that the SFT-based method achieves a better resolution than its RSFT counterpart, since the former does not require a pre-permutation window. For the second scenario, the SFT-based method shows the side-lobes of the stronger targets, while the RSFT-based method only recovers the (extended) main-lobes of all the targets. 

The simulation shows that the RSFT-based approach is better than its SFT counterpart for a realistic scenario, within which the signal has a reasonable dynamic range. We also want to emphasis that in a real radar system, determine the number of counting peaks for the SFT-based method lacks a theoretical foundation, while the thresholding approach in the RSFT is consistent with the conventional FFT-based processing, both of which are based on the NP criterion.

\begin{figure}[!t]
	\begin{center}
	\includegraphics[scale=0.43]{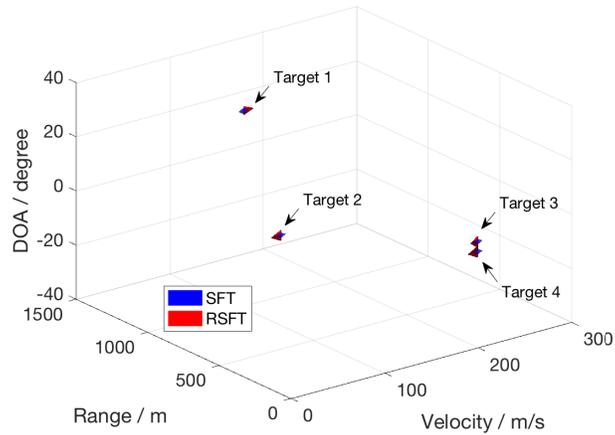} 
	\caption{\textbf{Target Reconstruction via 3-D SFT and RSFT with The Same SNR for 4 Targets.} Both SFT and RSFT based methods can reconstruct all the targets. The former has better resolution without pre-permutation windowing.}
	\label{fig:3DFFT}
	\end{center}
\end{figure}

\begin{figure}[!t]
	\begin{center}
	\includegraphics[scale=0.43]{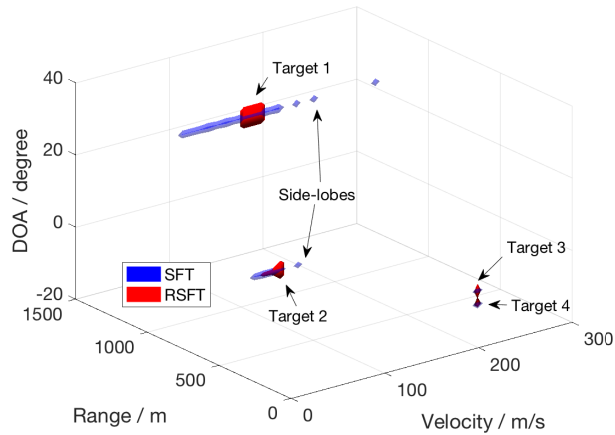} 
	\caption{\textbf{Target Reconstruction via 3-D SFT and RSFT with Different SNR for 4 Targets.} The SFT-based processing recovers the side-lobes of the stronger targets, while the RSFT-based method only recovers the main-lobes of targets. }
	\label{fig:3DSFT1}
	\end{center}
\end{figure}

\section{Conclusion} \label{sec:conclusion}
In this paper, we have addressed practical problems of applying SFT in real-life applications,  based on that, we have proposed a modified algorithm (RSFT). The optimal design of  parameters in RSFT has been analyzed, and the relationship between system sensitivity and computational complexity has been investigated. Some interesting properties of the RSFT have also been revealed by our analysis, such as the performance of detection not only relies on the frequency under examination, but also depends on other co-existing significant frequencies, which is very different from the traditional FFT-based processing. The analysis has revealed that RSFT could provide engineers an extra freedom of design in trading off system's ability of detecting weak signals and complexity. Finally, the context of the application of RSFT has been discussed, and a specific example  for short range ubiquitous radar signal processing has been presented.

\appendices
\section{Bartlett Method Analysis} \label{app:Bartlett}
Detecting each sinusoid and estimating the corresponding  frequencies in (\ref{eq:sig_model}) can be achieved by Bartlett spectrum analysis followed by an NP detection. A window is applied on each data segment to reduce the leakage of off-grid frequencies. In order to enhance the computational efficiency,  the FFT is adopted (see Algorithm \ref{alg:Bartlett}). In the $N$-dimensional case, each step in Algorithm \ref{alg:Bartlett} is done in $N$-dimension.

The analysis of the Bartlett method can be found in the literature \cite{stoica2005spectral}, however, such analysis usually focuses on bias, variance and frequency resolution, while the detection performance has not been throughly studied in connection with an NP detector. In \cite{so1999comparison}, the performance of detecting a single sinusoid is discussed and the theoretical analysis is provided for on-grid frequency. However, most typically, the signal contains multiple significant frequencies, which are off-grid. In what follows, we analyze the asymptotic  performance of Algorithm \ref{alg:Bartlett} based on the signal model of (\ref{eq:sig_model}), which is a multiple-frequencies case and does not restrict the frequency being on-grid.  Moreover, as compared to the signal model in  \cite{so1999comparison}, which assumes that the sinusoid has a deterministic amplitude, we model the complex amplitude of each sinusoid as a circularly symmetric Gaussian random variable.  This modeling reflects the stochastic nature of each sinusoid, and is consistent with the Swerling-\uppercase{I} target model in radar signal cases, since the square of a circularly symmetric Gaussian random variable has an exponential distribution.

\begin{algorithm}
\caption{Bartlett Method algorithm}\label{alg:Bartlett}
\textbf{Input:} complex signal $\mathbf{r}$ in any fixed dimension \\
\textbf{Output:} $\mathbf{o}$, frequency domain representation of input signal 

\begin{algorithmic}[1]
\Procedure{Bartlett}{$\mathbf{r}$} 
\State $\mathbf{x} \gets 0$
\For {$i=0 \to T$} 
	\State Windowing: $\mathbf{u} \gets \mathbf{W} \mathbf{r}$
	\State N-D FFT:  $\hat{\mathbf{u}}$ $\gets$ $\mathop{\mathrm{NDFFT}}(\mathbf{u})$
	\State Accumulation: $\mathbf{x} \gets \mathbf{x}+|\hat{\mathbf{u}}|^2$
\EndFor
         \State Detection: $\mathbf{o} \gets \mathop{\mathrm{NPdet}}(\mathbf{x})$ 
\State \textbf{return} $\mathbf{o}$
\EndProcedure
\end{algorithmic}
\end{algorithm}

The analysis of Algorithm \ref{alg:Bartlett} follows a similar fashion with the analysis of the RSFT, and we thus use the same notation as in Section \ref{sec:optimal}.  Our goal is to derive the relationship between $P_d$ and $P_{fa}$, which is also related to the worst case signal SNR, i.e., $SNR_{min}$.

After windowing and FFT, the signal becomes
\begin{equation} \label{eq:Bart_f}
\hat{\mathbf{u}}_s = \overline{\mathbf{D}}  \mathbf{W} \mathbf{r}_s,\;s \in[T],
\end{equation}
where $\overline{\mathbf{D}}  \in \mathbb{C}^{N \times N}$ is the DFT matrix.

Substituting (\ref{eq:sig_model}) into (\ref{eq:Bart_f}),  for the $k_{th}, k \in [N]$ entry of  $\hat{\mathbf{u}}_s$, we get
\begin{equation} \label{eq:f_km}
\begin{split}
[\hat{\mathbf{u}}_s]_{k} &= b_{m,s} \mathbf{d}^H_k \mathbf{W}  \mathbf{v} (\omega_m) \\
&+ \sum_{j \in [K]\setminus m} (b_{j,s} \mathbf{d}^H_k \mathbf{W} \mathbf{v} (\omega_j))\\
&+ \mathbf{d}^H_k \mathbf{W} \mathbf{n}_s,
\end{split}
\end{equation}
where  $\mathbf{d}_k$ is the $k_{th}$ column of $\overline{\mathbf{D}}$, i.e., $\mathbf{d}_k = [1\quad e^{j k \Delta \omega_N}\; \cdots\; e^{jk(N-1)\Delta \omega_N}]^T$, and $\Delta \omega_N = 2\pi/N$.
 
 Since $[\hat{\mathbf{u}}_s]_{k}$ is a linear combination of $b_{i,s}, [\mathbf{n}_s]_{j}, i\in[K], j\in[N]$, it is a circularly symmetry Gaussian scalar with distribution 
 \begin{equation} \label{eq:dist_f_k}
[\hat{\mathbf{u}}_s]_{k} \sim \mathcal{CN} (0, \sigma_{uk}^2). 
\end{equation}

The hypothesis test on each frequency bin is formulated as
\begin{itemize}
\item $\mathcal{H}'0$: no significant frequency exists.
\item $\mathcal{H}'1$: there exists a significant frequency, with its SNR at least equals to $SNR_{min}$.
\end{itemize}

We assume the side-lobes of the significant frequencies are far below the noise level due to windowing, then under $\mathcal{H}'1$ and $\mathcal{H}'0$, respectively, we have the following approximation for $\sigma_{uk}^2$

\begin{equation} 
\begin{split}
&\sigma_{uk|H1}^2 \approx \sigma_{bm}^2 \alpha' ,\\ 
&\sigma_{fk|H0}^2 \approx \sigma_n^2 \beta'. \\
\end{split}
\end{equation}
where $\alpha' = |\mathbf{d}^H_k \mathbf{W}  \mathbf{v} (\omega_m)|^2$ and $\beta' = \|\mathbf{w}\|^2$.

%

The LLRT yields the sufficient statistics
\begin{equation}
l_k = {1 \over T} \sum_{s\in [T]} |[\hat{\mathbf{u}}_s]_{k}|^2 \mathop{\gtrless}_{\mathcal{H}'0}^{\mathcal{H}'1} \gamma .
\end{equation}

We study its asymptotic performance. Assume that $T$ is moderately large, after applying central limit theory, the test statistic distributes as Normal distributions in both hypothesis, i.e.,
\begin{equation}
\begin{split}
&l_{k|H0} \sim \mathcal{N}(\sigma_{uk|H0}^2, {\sigma_{uk|H0}^4 \over T}) ,\\
&l_{k|H1} \sim \mathcal{N}(\sigma_{uk|H1}^2, {\sigma_{uk|H1}^4 \over T}) .
\end{split}
\end{equation}

Finally, we can relate $P_d$ and $P_{fa}$ as
\begin{equation} \label{eq:bartlettROC}
P_d = 1-\Phi \left( {\beta' \Phi^{-1}(1-P_{fa}) + \sqrt{T}(\beta' - SNR_{min} \alpha') \over SNR_{min} \alpha'} \right),
\end{equation}
where $\Phi(\cdot)$ is the CDF of standard normal distribution. An exemplar ROC curve calculated with (\ref{eq:bartlettROC}) is demonstrated in Fig. \ref{fig:ROC_B}.

\section{Proof of Properties of Mapping and Reverse Mapping} \label{app:map}
\subsection{Proof of Property \ref{prop:mapping_reversbility}}

\begin{proof}
According to Definition \ref{def:map}, the mapping can be split into two stages: 1) apply modular multiplication to $i$, i.e., $k = [\sigma i]_N \in [N]$; and, 2) convert $k$ into $j \in [B]$ with $j = \lfloor k B/N \rfloor$.

Similarly, according to Definition \ref{def:rev_map}, the reverse-mapping also can be split into two stages: 1) dilate $j \in [B]$ into $\mathbb{S} \equiv \{v \in[N] \mid   j\frac{N}{B} \le v < (j+1)\frac{N}{B}]\} \subset [N]$; and, 2) apply inverse modular multiplication on $\mathbb{S}$, i.e., $\mathcal{R}(j, \sigma^{-1}) \equiv \{ [u \sigma^{-1}]_N \mid u\in \mathbb{S} \}$.

The first stage of reverse-mapping is the inverse operation of the second stage of mapping, and as a result, $k \in \mathbb{S}$. Hence $i = [k \sigma^{-1}]_N \in \mathcal{R}(j, \sigma^{-1})$ as desired.
\end{proof}

\subsection{Proof of Property \ref{prop:mapping_exculusiveness}}

\begin{proof}
We use the two stages of the reverse-mapping in the proof of Property \ref{prop:mapping_reversbility}. The first stage of the reverse-mapping for $i$ and $j$ yields $\mathbb{S}_1 \equiv \{v \in[N] \mid  i\frac{N}{B} \le v < (i+1)\frac{N}{B}]\}$ and $\mathbb{S}_2 \equiv \{v \in[N] \mid j\frac{N}{B} \le v < (j+1)\frac{N}{B}]\}$, respectively. It is not difficult to verify that $\mathbb{S}_1 \cap \mathbb{S}_2 = \emptyset$, provided that $i \neq j$. 

In what follows, we will prove that the second stage of the reverse-mapping also gives distinct results.
Assume that there exists $m \in \mathbb{S}_1, n \in \mathbb{S}_2$, such that $[m \sigma^{-1}]_N = [n \sigma^{-1}]_N$. Modularly multiply both sides with $\sigma$ yields that $m = n$, which is contradictory with $\mathbb{S}_1 \cap \mathbb{S}_2 = \emptyset$. Hence both stages of the reverse-mapping guarantee the results are distinct for $i \neq j$.
\end{proof}

\section*{Acknowledgment}
The authors would like to thank Dr. Predrag Spasojevic and Dr. Anand Sarwate from Rutgers university for initial support of this work. The work of SW was jointly supported by China Scholarship Council and Shanghai Institute of Spaceflight Electronics Technology.  The work of VMP was partially supported by an ARO grant W911NF-16-1-0126.  

\bibliographystyle{ieeetr}
\bibliography{SFTbib}

%
%

\end{document}